\documentclass{article}
\usepackage{matharticle}
\usepackage{amsmath}
\usepackage{ifthen}
\usepackage{mdwlist}
\usepackage{bm}
\usepackage{xcolor}
\usepackage[colorlinks=true,linkcolor=blue,citecolor=blue,urlcolor=blue]{hyperref}
\usepackage{enumitem}
\usepackage{graphicx}
\usepackage{xspace}
\usepackage{verbatim}
\usepackage[margin=1in]{geometry}
\usepackage{thm-restate}
\usepackage{thmtools}

\usepackage{latexsym}
\usepackage{epsfig}
\usepackage{mleftright}
\usepackage[colorinlistoftodos,textsize=scriptsize]{todonotes}
\def\withcolors{0}
\def\withnotes{0}

\renewcommand{\epsilon}{\ve}
\def\ve{\varepsilon}

\newcommand{\pr}[2][]{\mathrm{Pr}\ifthenelse{\not\equal{}{#1}}{_{#1}}{}\!\left[#2\right]}

\providecommand{\poly}{\operatorname*{poly}}
\newcommand{\id}{\mathbb{I}}

\newcommand{\expectation}[1]{\mathbb{E}\left[#1\right]}
\newcommand{\absv}[1]{\left |#1\right |}
\newcommand{\probof}[1]{\Pr\Paren{#1}}
\newcommand{\variance}[1]{\Var\Paren{#1}}
\def \Bin{Bin}

\def \covariance {\Cov}

\ifnum\withcolors=1
  \newcommand{\gcolor}[1]{{\color{red}#1}}
  \newcommand{\acolor}[1]{{\color{purple}#1}}
\else

  \newcommand{\gcolor}[1]{{#1}}
\fi

\ifnum\withnotes=1
  \newcommand{\gnote}[1]{\par\gcolor{\textbf{G: }\sf #1}} 
  \newcommand{\gfootnote}[1]{\footnote{{\bf \gcolor{Gautam}}: {#1}}}
  \newcommand{\anote}[1]{\par\acolor{\textbf{J: }\sf #1}} 

\else
  \newcommand{\gnote}[1]{}
  \newcommand{\anote}[1]{}
  \newcommand{\gfootnote}[1]{}

\fi
\newcommand{\ignore}[1]{\leavevmode\unskip} 

\newcommand{\AdjM}{A}
\newcommand{\Bounda}{\kappa}
\newcommand{\Boundb}{\eta}
\newcommand{\regpar}{\alpha}

\usepackage{soul}
\newcommand{\tcr}[1]{#1}

\newcommand{\jd}[1]{\textcolor{purple}{#1}}
\newcommand{\newhz}[1]{#1}

\def \ns{n}

\newcommand{\gnp}[2]{G(#1,#2)}

\newcommand{\corr}{\gamma}

\newcommand{\erp}[1]{\mathcal{E}\Paren{#1}}

\newcommand{\adver}{\mathcal{A}}

\newcommand{\bad}{B}

\newcommand{\setsize}{C}

\newcommand{\goodnodes}{F}
\newcommand{\goodnodesub}{\goodnodes'}
\newcommand{\goodnodesubanother}{\goodnodes''}

\title{Robust Estimation for Random Graphs}
\author{
Jayadev Acharya\thanks{School of Electrical and Computer Engineering, Cornell University. {\tt acharya@cornell.edu}. Supported by NSF-CCF-1846300 (CAREER),
    NSF-CCF-1815893, and a Google Faculty Research Award.}
    \and
Ayush Jain\thanks{School of Electrical and Computer Engineering. {\tt  ayjain@eng.ucsd.edu}. Supported by NSF-CCF-1564355 and NSF-CCF-1619448.}
\and
Gautam Kamath\thanks{Cheriton School of Computer Science, University of Waterloo. {\tt g@csail.mit.edu}. Supported by an NSERC Discovery Grant and a University of Waterloo startup grant.}
\and
Ananda Theertha Suresh\thanks{Google Research, New York. {\tt theertha@google.com}}
\and
Huanyu Zhang\thanks{Meta. Most of this work was done while the author was a graduate student at Cornell University. {\tt hz388@cornell.edu}. Supported by NSF-CCF-1815893.}

}

\begin{document}

\maketitle
\thispagestyle{empty}

\begin{abstract}
We study the problem of robustly estimating the parameter $p$ of an Erd\H{o}s-R\'enyi random graph on $n$ nodes, where a $\gamma$ fraction of nodes may be adversarially corrupted.
  After showing the deficiencies of canonical estimators, we design a computationally-efficient spectral algorithm which estimates $p$ up to accuracy $\tilde O(\sqrt{p(1-p)}/n + \gamma\sqrt{p(1-p)} /\sqrt{n}+ \gamma/n)$ for $\corr < 1/60$.
  Furthermore, we give an inefficient algorithm with similar accuracy for all $\corr <1/2$, the information-theoretic limit.
  Finally, we prove a nearly-matching statistical lower bound, showing that the error of our algorithms is optimal up to logarithmic factors.
\end{abstract}

\newpage
\section{Introduction}
\setcounter{page}{1}
Finding underlying patterns and structure in data is a central task in machine learning and statistics. Typically, such structures are induced by modelling assumptions on the data generating procedure.
While they offer mathematical convenience, real data generally does not match with these idealized models, for reasons ranging from model misspecification to adversarial data poisoning. 
Thus for learning algorithms to be effective in the wild, we require methods that are \emph{robust} to deviations from the assumed model.

With this motivation, we initiate the study of robust estimation for random graph models.
Specifically, we will be concerned with the
Erd\H{o}s-R\'enyi (ER) random graph model~\cite{Gilbert59,ErdosR59}.\footnote{This model was introduced by Gilbert~\cite{Gilbert59}, simultaneously with the related $G(n,m)$ model of Erd\H{o}s and R\'enyi~\cite{ErdosR59}. Nevertheless, the community refers to both models Erd\H{o}s-R\'enyi graphs.}
\begin{definition}[Erd\H{o}s-R\'enyi graphs]
  \label{def:gnp}
The Erd\H{o}s-R\'enyi random graph model on $n$ nodes with parameter $p \in [0, 1]$, denoted as  $\gnp{n}{p}$, is the distribution over graphs on $n$ nodes where each edge is present with probability $p$, independently of the other edges.
\end{definition}

We consider graphs generated according to the Erd\H{o}s-R\'enyi random graph model, but which then have a constant fraction of their nodes \emph{corrupted} by an adversary. When a node is corrupted, the adversary can arbitrarily modify its neighborhood.  
This setting is naturally motivated by social networks, where random graphs are a common modelling assumption~\cite{NewmanWS02}.
Even if a fraction of individuals in the network are malicious actors, we still wish to perform inference with respect to the regular users.
Apart from adversarial settings, tools for robust analysis of graphs may also assist in addressing deficiencies of existing models, such as in model misspecification.
For example, certain random graph models have  been criticized for not capturing various statistics of real-world networks~\cite{NewmanWS02}, and some notion of robustness may facilitate better modelling.


%

\subsection{Problem Setup} Let $\corr\in[0,1]$ denote the fraction of corrupted nodes, and $G\sim\gnp{n}{p}$ be a random graph, where $p$ is unknown. Without loss of generality, we assume that the node set is $[n]:=\{1,\ldots,n\}$.
An adversary $\adver$ is then given $G$, and is allowed to arbitrarily `rewire' the edges adjacent to a set $B\subseteq[n]$ of nodes of size at most $\corr n$, resulting in a graph $\adver(G)$. In other words, the adversary can change the status of any edge with at least one end point in $B$. We call $B$ the set of \emph{corrupted nodes}.
We consider two kinds of adversaries. 
\begin{itemize}
	\item ~$\corr$-\emph{omniscient adversary}: The adversary knows the true value of the edge probability $p$ and observes the realization of the graph $G\sim G(n,p)$. They then choose $B$ and how to rewire its edges. 
\item ~$\corr$-\emph{oblivious adversary}: The adversary knows the true value of the edge probability $p$.
They must choose $B$ and the distribution of edges from $B$ without knowing the realization $G$. 
\end{itemize}
Note that the oblivious adversary is weaker than the omniscient adversary. Given a corrupted graph $\adver(G)$, our goal is to output $\hat p
\left(\adver(G)\right)$, an estimate of the true edge probability $p$.

\subsection{Results}
\label{sec:Results}
We first analyze standard estimators from the robust statistics toolkit, and show that they provide sub-optimal rates. We then propose a computationally-efficient spectral algorithm  to estimate $p$ with improved rates.
Finally, we prove a lower bound for this problem, showing that our algorithms are optimal up to logarithmic factors. We note that our upper bounds hold for omniscient adversaries, whereas the lower bounds are tight even against the weaker oblivious adversary.

\subsubsection{Standard Robust Estimators and  Natural Variants}
At first glance, the problem appears deceptively simple, as our goal is to estimate a single univariate parameter $p$. 
A standard technique is the maximum likelihood estimator, which in this case is the empirical edge density. We call the following the \emph{mean estimator}
 \begin{align}
 \hat{p}_{\rm mean}(\adver(G)) = \frac{\text{\# of edges present in $\adver{(G)}$}}{{n \choose 2}}.
 \label{eqn:empirical-mean-estimator}
 \end{align}
 
In robust statistics, the median often provides better guarantees than the mean. Let $\deg(i)$ denote the degree of node $i\in[n]$ in $\adver(G)$. The \emph{median estimator} is given by
\begin{align}
 \hat{p}_{\rm med}(\adver(G)) = \frac{{\rm Median}\{\deg(1),\ldots, \deg(n)\}}{{n-1}}.
 \label{eqn:empirical-median-estimator}
 \end{align}

Absent corruptions (i.e., $\corr = 0$), we have $\adver(G)=G$.
In this simple setting, the mean and median are both very accurate.
Specifically, it is not hard to show that $\left|\hat p_{\rm mean}(G) - p\right| \le  O\left(\sqrt{p(1-p)}/n\right)$ 
 and  $\left|\hat p_{\rm med}(G) - p\right| \le O \left(1/n\right)$ (Lemma~\ref{lem:without-corruptions}).
However, both estimators perform much worse under even mild corruption.
In Lemma~\ref{lem:with-corruptions} we describe and analyze a simple oblivious adversary $\adver$ such that both the mean and median estimator have $\left|\hat p(\adver(G)) - p\right| \ge \corr/2$. 
Note that if even a single node is corrupted (i.e., $\corr=1/n$), the ``price of robustness'' (informally, the additional error term(s) introduced in the corrupted setting)  dominates the baseline $O(1/n)$ error in the uncorrupted setting.

The adversary against the mean and median estimators is easy to describe: either add or remove all edges incident to the nodes in $B$.
This suggests the strategy of first pruning a set of $c\gamma n$ nodes with the largest and smallest degrees and then applying either the mean or median estimator to the resulting graph.
These \emph{prune-then-mean/median} algorithms are described in Algorithm~\ref{algo:prune}.
Despite this additional step, the pruned estimators are still deficient.
We design an oblivious adversary such that the prune-then-median estimate satisfies $\left|\hat p(\adver(G)) - p\right| \ge \Omega (\corr)$ and the prune-then-mean estimate satisfies  $\left|\hat p(\adver(G)) - p\right| \ge \Omega(\corr^2)$ (Theorem~\ref{thm:prune-median-lb}).
Interestingly, we show the tightness of both these bounds, showing that prune-then-mean improves the error to $O\Paren{\gamma^2}$ (Theorem~\ref{thm:prune-median-ub}). 
These results are summarized in the theorem below.  
\begin{theorem}[Informal]
The price of robustness of the prune-then-mean/median estimators are $\Theta(\corr^2)$ and $\Theta(\corr)$, respectively.
\end{theorem}

\subsubsection{A Spectral Algorithm for Robust Estimation}
 
Given the failings of the approaches described so far, it may appear that a $\poly(\corr)$ cost for robustness may be unavoidable. Our main result is a computationally-efficient algorithm that bypasses this barrier.

\begin{theorem} 
\label{thm:main}
Suppose $\corr<1/60$ and $p \in [0,1]$. Let $G\sim G(n,p)$ 
and $\adver(G)$ be a rewiring of $G$ by a $\corr$-omniscient adversary $\adver$. There exists a polynomial-time estimator $\hat p (\adver(G))$ such that with probability at least $1-{10}n^{-2}$,
 \begin{align}
|\hat p (\adver(G))-p| \le C\cdot\Paren{ \frac{\sqrt{p(1-p)\log n}}{n}+\frac{\corr\sqrt{p(1-p)\log(1/\gamma)}}{\sqrt{n}}+  \frac\corr n\log n },
 \end{align}
 for some constant $C$.
 This {estimate} can be computed in $\tilde{O}(\gamma n^3+n^{2})$ time. 
\end{theorem}

The first term is the error without corruptions, while the other two terms capture the price of robustness. 
Except at extreme values of $p$, the last term will be dominated by one of the other two. In this case, note that 
the cost of robustness in the second term decreases as the number of nodes $n$ increases. 
This is in contrast to the previously described approaches, for which the price of robustness did not decrease with $n$.
Observe that the non-robust error will dominate for most regimes when $\corr \leq 1/\sqrt{n}$.

As our lower bounds will establish, our algorithm provides a nearly-tight solution to the problem. 
Note that while this algorithm requires knowledge of $\corr$, Jain, Orlitsky, and Ravindrakumar~\cite{JainOR22} recently proposed a simple argument which generically removes the need to know the corruption parameter for robust estimation tasks, leading to such an algorithm with the same rates.


Our upper bound requires $\gamma<1/60$.\footnote{We have not tried to optimize the value of $\gamma$ for computationally efficient algorithms, and  could likely be made larger than $1/60$ through a more careful analysis.} On the other hand, note that if $\gamma\ge 0.5$, an identifiability argument implies that no estimator can achieve error better than $0.5$.\footnote{Consider an empty graph $G(n,0)$. An adversary can corrupt half the graph into a clique, making it look like it came from $G(n,1)$. No algorithm can identify which half of the graph was the original.} 
This raises the question of whether the above rates are achievable for all $\corr<0.5$. 
We show that this is indeed the case, providing a computationally inefficient algorithm with the following guarantees.


\begin{theorem}
\label{thm:large-gamma}
Suppose $\gamma<1/2$. There exists an algorithm such that with probability at least $1-n^{-2}$,
 \begin{align}
|\hat p (\adver(G))-p| \le \frac{C}{1/2 - \gamma}\cdot\Paren{\frac{\sqrt{p(1-p)}}{\sqrt{n}}+  \frac {\sqrt{\log n}}{n} },
 \end{align}
 for some constant $C$. 
\end{theorem}
Note that for $\corr>1/60,$ the error bound above matches that presented in Theorem~\ref{thm:main} up to a factor of $1/(0.5-\corr)$, and therefore extends the error rates of Theorem~\ref{thm:main} to the regime $\corr \in [1/60, 1/2)$ at the cost of computational efficiency.

\subsubsection{Information-theoretic Lower Bounds}

We provide a lower bound to establish  near-optimality of our algorithms. While our upper bounds are against an omniscient adversary, the lower bounds hold for the weaker oblivious adversary.
{
\begin{theorem}
\label{thm:lower-bound-main}
For every $\gamma<1/2$, $p\in[0,1]$, and $n\ge 0$ and a universal constant $C'$, let 
\[
\Delta = C'\cdot\Paren{ \frac{\sqrt{p(1-p)}}{n}+\frac{\corr\sqrt{p(1-p)}}{\sqrt{n}}+  \frac\corr n}.
\]
For any $p'\in [p+\Delta,p-\Delta]$ and $G\sim G(n,p)$ and $G'\sim G(n,p')$, there exists an oblivious adversary $\adver$ such that no algorithm can distinguish between $\adver(G)$ and $\adver(G')$ with probability more than 0.65.
\end{theorem}}

\subsection{Techniques}
\label{sec:techniques}



\medskip
\noindent\textbf{{Upper bound techniques.}}
Broadly speaking, robust estimation is only possible when samples from the (uncorrupted) distribution enjoy some nice structure. 
Work in this area generally proceedings by imposing some regularity conditions on the uncorrupted data, which hold with high probability over samples from the distribution. 
The algorithm subsequently relies solely on these regularity conditions to make progress.
For example, for mean estimation problems, it is common to assume that the mean and covariance of the uncorrupted samples are close to the true mean and covariance.
However, the appropriate regularity conditions in our setting are far less obvious. 
We employ conditions which bound the empirical edge density and spectral norm for submatrices of the adjacency matrix, when appropriately centered around the true parameter $p$ (Definition~\ref{def:regular}), which can be proven using tools from random matrix theory.


With our regularity conditions established, the algorithmic procedure proceeds in two stages: a coarse estimator, followed by a fine estimator.

\medskip
\noindent\textbf{Stage 1: A coarse estimate.}
Our regularity conditions are suggestive of the following intuition about how one might estimate the value of $p$.
If one could locate a sufficiently large subgraph $S$ of the uncorrupted nodes, such that their adjacency matrix centered around $p$ has small spectral norm, then the empirical edge density of this subgraph would give a good estimate for the true parameter $p$.
More precisely, we let $A$ be the (corrupted) adjacency matrix of $\adver(G)$, let $A_{S\times S}$ be the submatrix of $A$ indexed by the set $S$, and $p_S$ be the empirical edge density of the subgraph $S$.
The goal is to obtain an $S$ where $\|A_{S \times S} -p_S \|$ is small,\footnote{For clarity: in the expression $A_{S \times S} -p$, $p$ is subtracted entry-wise.} at which point we can output $p_S$.

There are two clear challenges with this approach. 
First off, we can not center the adjacency matrix around the unknown parameter $p$, since estimating that parameter is our goal.
However, we demonstrate that it instead suffices to center around $p_S$ (Theorem~\ref{th:newth}).
The other issue is that it is not clear how to identify such a set $S$ of uncorrupted nodes.
One (inefficient) approach is to simply inspect all sufficiently large subgraphs.
This will be accurate (quantified in Theorem~\ref{th:ineff}), but not computationally tractable.


Instead, our main algorithmic contribution is an efficient algorithm which achieves this same goal. 
We give an iterative spectral approach, which starts with $S = [n]$.
In Lemma~\ref{lem:delprob} we show that if the spectral norm of $A_{S\times S}-p_S$ is large, then the top eigenvector assigns significant weight to the set of corrupted nodes.
Normalizing this eigenvector and sampling from the corresponding probability distribution identifies a corrupted node with constant probability.
We eliminate this node from $S$ and repeat the process. Finally, using this approach, we obtain a subset $S^*\subset[n]$ of nodes such that $p_{S^*}$ is a coarse estimate of $p$.

\medskip
\noindent\textbf{Stage 2: Pruning the coarse estimate.}
It turns out that the above coarse estimate gives a price of robustness which is roughly $O(1/\sqrt{n})$, rather than the $O(\corr/\sqrt{n})$ we are trying to achieve. 
However, a simple pruning step allows us to complete the argument.
Specifically, our coarse estimator gave us a set $S^*$ such that the spectral norm of $A_{S^*\times S^*}-p_{S^*}$ is small and $p_{S^*}$ is close to $p$.
We employ this to show that most nodes must have degree close to $p$ (Lemma~\ref{lem:trimnodes}).
Thus, we remove $\Theta(\corr n)$ nodes whose degree (restricted to the subgraph $S^*$) is furthest from $p_{S^*}$.
Our final estimate is the empirical density of the resulting pruned subgraph.




\medskip
\noindent\textbf{{Lower bound techniques.}} A strategy for proving lower bounds is the following: Suppose there exists an adversary that with $\gamma n$ corruptions can convert the distribution $G(n, p)$ and $G(n,p+\delta)$ into the same distribution of random graphs, then we cannot estimate $p$ to accuracy better than $\delta/2$. This is akin to couplings between $G(n, p)$ and $G(n,p+\delta)$ by corrupting only a $\corr n$ nodes. Designing these couplings over  Erd\H{o}s-R\'enyi graphs can be tricky due to the fact that degrees of nodes are not independent of each other. 

We instead consider directed Erd\H{o}s-R\'enyi graphs, where an edge from a node $i$ to $j$ is present independently of all others. Then, the (outgoing) degrees of all the nodes are independent Binomial distributions. Using total variation bounds between Binomial distributions we can design couplings between directed ER graphs with different parameters, thus showing a lower bound on the error of robustly estimating directed ER graphs. Our final argument is a reduction showing that estimating  the parameters of undirected of graphs is at least as hard as estimating the parameters of directed ER graphs. Combining these bounds we obtain the lower bounds.

\subsection{Related Work}
\label{sec:related}

Robust statistics is a classic and mature branch of statistics which focuses on precisely this type of setting since at least the 1960s~\cite{Tukey60,Huber64}. 
However, since the classic literature typically did not take into account computational considerations, proposed estimators were generally intractable for settings of even moderate dimensionality~\cite{Bernholt06}.
Recently, results by Diakonikolas, Kamath, Kane, Li, Moitra, and Stewart~\cite{DiakonikolasKKLMS16,DiakonikolasKKLMS19} and Lai, Rao, and Vempala~\cite{LaiRV16} overcame this barrier, producing the first algorithms which are both accurate and computationally efficient for robust estimation in multivariate settings.
While they focused primarily on parameter estimation of Gaussian data, a flurry of subsequent works have provided efficient and accurate robust algorithms for a vast array of settings, see, e.g.,~\cite{DiakonikolasKKLMS17,CharikarSV17,BalakrishnanDLS17,KlivansKM18,SteinhardtCV18,DiakonikolasKKLMS18,HopkinsL18,KothariSS18,DiakonikolasKKLSS19,ChengDG19,ChengDGW19,DongHL19,ZhuJS19,ChengDGS20,PrasadSBR20a,PensiaJL21,LiuSLC20}, or~\cite{DiakonikolasK19} for a survey.

A common tool in several of these on robust estimation results is to prune suspected outliers from the dataset so that a natural estimator over the remaining points has a small error. 
We also use this meta technique in this paper. We note that as in the previous works, the main challenge lies in designing efficient schemes to detect and remove corrupted data-points for the particular task at hand. In most prior works, the uncorrupted data-points are unaffected by corruptions. In our setting however, the edges from the good nodes are also affected by corruptions to the corrupted nodes. This presents an new challenge requiring new insights.

Specific outlier-removal techniques vary throughout the literature on robust estimation, but most are spectral in nature.
Prior to this line of work, similar approaches were employed for robust supervised learning tasks, namely learning halfspaces with malicious noise~\cite{KlivansLS09, AwasthiBL14}.

Some prior works have studied robust estimation for graphical models, including Ising models~\cite{LindgrenSSDK18,PrasadSBR20b} and Bayesian networks~\cite{ChengDKS18}.
Despite the common nomenclature, these works are rather different from our work on random graph models.
Graphical models are distributions over vectors, where correlations between coordinates exist based on some latent graph structure.
On the other hand, random graph models are distributions over graphs, sampled according to some underlying parameters.
While existing work on graphical models necessitates many samples from the same distribution (due to parameters outnumbering the samples), our setting requires a single sample from a random graph model.

Our setting is related to the untrusted batches setting of Qiao and Valiant~\cite{QiaoV18}, in which many batches of samples are drawn from a distribution, but a constant fraction of batches may be adversarially corrupted, see also followup works by Jain and Orlitsky~\cite{JainO20a,JainO20b,JainO21} and Chen, Li, and Moitra~\cite{ChenLM20a,ChenLM20b}.
This is somewhat similar to our setting, where each batch is the set of edges connected to a node.
However, the key difference is that in our setting, each edge belongs to both its two endpoint nodes, whereas in the untrusted batches setting, a sample is only associated with a single batch.

Estimation in random graph models has also been studied under the constraint of differential privacy~\cite{BorgsCS15,BorgsCSZ18a,SealfonU19}.
Despite superficial similarities between the two settings, we are unaware of deeper technical connections.

There has been significant work on robust community detection in the presence of adversaries~\cite{MoitraPW16, MakarychevMV16, SteinhardtCV18, BanksMR21}.
Most of this focuses on monotone adversaries (which make only ``helpful'' changes to the graph) or edge corruptions.
It is not clear how to define monotone adversaries for the Erd\H{o}s-R\'enyi setting, and for our estimation problem under edge corruptions, the empirical estimator is trivially optimal in the worst case.
The work of Cai and Li~\cite{CaiL15} also considers a node corruption model similar to ours.
However, all of the aforementioned work studies community detection in stochastic block models, which is different from our goal of parameter estimation.

Our corruption model may seem reminiscent of the classic planted clique problem~\cite{Karp75, Jerrum92, Kucera95}, in which an algorithm must distinguish between a) $G(n,1/2)$ and b) $G(n,1/2)$ with the addition of a planted clique of size $\corr n$. 
Our adversary is given much more power (i.e., they can make arbitrary changes to the neighbourhoods of their selected nodes), though the two goals are incomparable.
The planted clique problem is known to be information-theoretically solvable for any $\corr > \frac{2\log n}{n}$.
However, polynomial-time algorithms are only known for $\corr > 1/\sqrt{n}$~\cite{AlonKS98}, and there is strong evidence that efficient algorithms do not exist for smaller values of $\corr$~\cite{FeigeK03, FeldmanGRVX17, MekaPW15, DeshpandeM15, HopkinsKPRS18, BarakHKKMP19}.
We have not run into issues in our setting related to this intractability, though deeper connections between our model and the planted clique problem would be interesting. 
Note that our task of parameter estimation is not interesting for the cases of the planted clique problem when $\corr \leq 1/\sqrt{n}$.
Simply using the empirical estimator on the two instances would give error $\approx 1/n$ and $\approx 1/n + \gamma^2 = O(1/n)$, which are identical up to constant factors.

Our setting bears some conceptual similarity to a line of robustness work focused on decomposing a matrix as a sum of a low rank matrix and a sparse matrix~\cite{ChandrasekaranSPW11, CandesLMW11, HsuKZ11}.
Our true parameter matrix is the rank-1 matrix $pJ$, where $J$ is the all-ones matrix.
However, the uncorrupted adjacency matrix is a sample from the distribution where each entry is a Bernoulli with the corresponding parameter, which is in general not low rank.
Furthermore, our corruption model allows for a bounded number of rows/columns to be changed, whereas this line of work requires that the corruptions satisfy some further sparsity, such as a limited number of changed entries per row/column, or that the corruption positions are chosen randomly.

\section{Notation and Preliminaries}

\label{sec:notations}
\paragraph{Problem Formulation.} Let $G\sim G(n,p)$. An adversary observes $G$  and chooses a subset $\bad\subseteq[n]$ of nodes with $|\bad|\le\corr n$. It can then change the status (i.e., presence or non-presence) of any edge with at least one node in $B$ to get a graph $\adver(G)$. Let $\goodnodes=[n]\setminus \bad$. 
We call $\bad$ the \emph{corrupted} nodes, and $\goodnodes$ the \emph{uncorrupted} nodes.
Let $\tilde{A}$ and $A$ be the $n\times n$ adjacency matrix of the original graph $G$ and the modified graph $\adver(G)$ respectively. 
Then $A_{\goodnodes\times \goodnodes} = \tilde A_{\goodnodes\times \goodnodes}$ and the remaining entries of $A$ can be arbitrary. 
Given $A$, the goal is to estimate $p$, the parameter of the underlying random graph model. The algorithm does not know the set $\bad$, though we assume that it knows the value of $\corr$.

\paragraph{Notation.}
The $\ell_2$ norm of a vector $v = [v_1,\ldots,v_n]\in \mathbb R^n$ is $\|v\|:= \sqrt{\sum_{i=1}^n v_i^2}$.
Suppose $M$ is an $m\times n$ real matrix. The spectral norm of $M$ is 
\begin{align}
\norm{M} := \max_{u\in\mathbb R^m, v\in \mathbb R^n: \norm{u}=1, \norm{v}=1}|{u^T M v}|.
\label{eqn:spec-norm}
\end{align}
It is easy to check that $\norm{M}= \max_{v\in \mathbb R^n: \norm{v}=1}\norm{ M v}$.
For a matrix $M$ and real number $a\in \mathbb R$, let $M-a$ be the matrix obtained by subtracting $a$ from each entry of $M$.
For ${S}\subseteq[m]$, $S'\subseteq [n]$, let $M_{S\times S'}$ be the $m\times n$ matrix that agrees with $M$ on $S\times S'$ and is zero elsewhere.
Similarly for a vector $v\in \mathbb R^n$ and $S\subseteq[n]$, let vector $v_S$ be the vector that agrees with with $v$ on $S$ and has zero entries elsewhere.


\paragraph{Matrix Properties.} Our algorithm will use several standard properties of the matrix spectral norm, which we state for completeness.

\begin{lemma}\label{lem:triangleineq}
Let $M, M' \in  \mathbb R^{m\times n} $, then
$
\norm{M+M'}
 \le
\norm{M}+\norm{M'}.
$
\end{lemma}

\begin{lemma}\label{lem:spsubmat}
For any $M \in  \mathbb R^{m\times n}$, $S\subseteq[m]$, $S'\subseteq [n]$, 
$
\|M_{S\times S'}\| \leq \|M\|$.
\end{lemma}
\begin{proof}
\tcr{For any unit vectors $u\in \mathbb R^m$ and $v\in \mathbb R^n$, let $\tilde u = u_S/\|u_S\|$ and $\tilde v = v_{S'}/\|v_{S'}\|$. Then 
\[|u^\intercal M_{S\times S'}v| = |u^\intercal_{S} M v_{S'}|= \|u_S\|\cdot\|v_{S'}\|\cdot|\tilde u^\intercal M\tilde v|\le |\tilde u^\intercal M\tilde v|\le \|M\|,\]
where the second last inequality used $\|u_S\|\le \|u\|=1$ and $\|v_{S'}\|\le \|v\|=1$ and the last inequality used that $\tilde u$ and $\tilde v$ are unit vectors. Finally, in the above equation taking maximum over all unit vectors $u,v$ completes the proof.
}
\end{proof}

\begin{lemma} \label{lem:specnormlb}
For any $M \in  \mathbb R^{m\times n}$,
$
\|M\| \geq \frac{\left|\sum_{i,j}M_{i,j}\right|}{\sqrt{mn}}.
$
\end{lemma}
\begin{proof}
Consider $ u = \frac{1}{\sqrt{m}}[1,1,\dots,1]$ and $ v = \frac{1}{\sqrt{n}}[1,1,\dots,1]^T$, which are unit vectors in $\mathbb R^{m}$ and $\mathbb R^{n}$, respectively. Then $|{u^T M v}| =  \frac{|\sum_{i,j}M_{i,j}|}{\sqrt{mn}}\le \norm{M}$ by~\eqref{eqn:spec-norm}.
\end{proof}


\section{Mean- and Median-based Algorithms}
To demonstrate the need for our more sophisticated algorithms in Section~\ref{sec:upper-bounds}, we first analyze canonical robust estimators for univariate settings -- specifically, approaches based on trimming and order statistics (i.e., the median). 

Recall the mean and median estimators for $p$ in~\eqref{eqn:empirical-mean-estimator} and~\eqref{eqn:empirical-median-estimator}.
The following simple lemma quantifies their guarantees in the setting absent corruptions.

\begin{lemma}
Suppose $\corr=0$. There exists a constant $C>0$ such that with probability at least 0.99,
$\absv{\hat p_{\rm mean} (G) - p} \le C \cdot \frac{\sqrt{p(1-p)}}n,  \ \text{and} \ \ \absv{\hat p_{\rm med} (G) - p} \le C \cdot \frac1n$.
\label{lem:without-corruptions}
\end{lemma}

The analysis of these estimators is not difficult, but we include them for completeness in Section~\ref{sec:proof-without-corruptions}.
Analysis of the median estimator is slightly more involved due to correlations between nodes.


While both estimators are optimal up to constant factors (for constant $p$) without corruptions, their performance decays rapidly in the presence of an adversary, scaling at least linearly in the corruption fraction $\corr$. In particular, consider an adversary that picks $\corr n$ nodes at random and either adds all the edges with at least one endpoint in $B$ or removes all of them. In Section~\ref{sec:mean-median-with-corruptions} we prove the following lower bound on the performance of the mean and median estimators for such an adversary.
Observe that if even one node is corrupted (i.e., $\corr\ge 1/n$), the error in Lemma~\ref{lem:with-corruptions} dominates the error without corruptions in Lemma~\ref{lem:without-corruptions}. 

\begin{lemma}
There exists an adversary $\adver$ such that for $\hat p\in \{\hat p_{\rm mean}(\adver(G)), \hat p_{\rm med}(\adver(G))\}$ with probability at least 0.5, we have $\absv{\hat p - p} \ge \corr/2.$
\label{lem:with-corruptions}
\end{lemma}


A common strategy in robust statistics is to prune or trim the most extreme outliers.
Accordingly, in our setting, one may prune the nodes with the most extreme degrees, described in Algorithm~\ref{algo:prune}.
This strategy bypasses the adversary which provides the lower bound in Lemma~\ref{lem:with-corruptions}.


\begin{algorithm}[H]
  \begin{algorithmic}
    \Require{A graph $\adver(G)$, corruption parameter $\corr$, a constant $c>0$}
\State{Remove $c\corr n$ nodes with largest and smallest degrees from $\adver(G)$}
\State{Apply the mean/median estimator from~\eqref{eqn:empirical-mean-estimator}/\eqref{eqn:empirical-median-estimator} to the resulting graph on $(1-2c\gamma)\cdot n$ nodes}
  \end{algorithmic}
  \caption{\label{algo:prune}Prune-then-mean/median algorithm}
\end{algorithm}

However, this strategy can only go so far. 
Roughly speaking, pruning improves the mean's robust accuracy from $\Theta(\corr)$ to $\Theta(\corr^2)$, while pruning does not improve the median's robust accuracy.
The upper and lower bounds are described in Theorems~\ref{thm:prune-median-ub} and~\ref{thm:prune-median-lb}, and proved in Sections~\ref{sec:prune-ub} and~\ref{sec:prune-lb}, respectively.

\begin{restatable}{theorem}{prunethenmedianub}
\label{thm:prune-median-ub}
For $c\ge 1$ and $0<\gamma \cdot c<0.25$, the prune-then-mean and prune-then-median estimators described in Algorithm~\ref{algo:prune} prune $2c\gamma n$ nodes in total and with probability $1-n^{-2}$ estimates $p$ to an accuracy $\mathcal O\bigl(c\corr^2+\frac{\log n}{n}\bigr)$ and $\mathcal O\bigl(c\gamma+\sqrt{\frac{\log n}{n}} \bigr)$, respectively.
\end{restatable}

\begin{restatable}{theorem}{prunethenmedianlb}
\label{thm:prune-median-lb}
Let $p=0.5$,  $\gamma>100\cdot\sqrt{{\log n}/{n}}$, and $c>0$ be such that $c\gamma<0.25$. There exists an adversary such that with probability at least 0.99, the prune-then-median estimate  that deletes $c\gamma n$ satisfies $\left|\hat p(\adver(G)) - p\right|\ge C'\gamma$, and the prune-then-mean estimate  satisfies $\left|\hat p(\adver(G)) - p\right|\ge C'\gamma^2$.
\end{restatable}

To summarize: none of the standard univariate robust estimators we have explored are able to achieve error better than $\Omega(\corr^2)$.
To bypass this barrier, we turn to more intricate techniques in designing our main estimator in Section~\ref{sec:upper-bounds}.



\section{An Algorithm for Robust Estimation}
\label{sec:upper-bounds}

Non-trivial robust estimation in Erd\H{o}s-R\'enyi graphs is possible because even if the set of edges connected to a small set of nodes is changed arbitrarily, the subgraph between the remaining nodes retains a certain structure.
In Section~\ref{sec:regcon}, we formalize this structure as deterministic regularity conditions and show that the subgraph corresponding to the set of uncorrupted nodes satisfy these regularity conditions with high probability.
In the following subsections, we use only the fact that that the subgraph of the uncorrupted nodes satisfy these regularity conditions to derive our robust algorithms for estimating $p$.

In Section~\ref{sec:conseqreg}, we first derive a simple inefficient spectral algorithm for coarse estimation of $p$.
Our efficient algorithm consists of two parts: an efficient version of the spectral algorithm in Section~\ref{sec:conseqreg} that, as its inefficient counterpart, provides a coarse estimate of $p$, followed by a trimming algorithm which achieves near-optimal error rates for estimating $p$.
We describe and analyze the spectral and trimming components of the algorithm in Sections~\ref{sec:upper-spectral} and~\ref{sec:upper-trimming}, respectively.
Finally, we put the pieces together to achieve our final upper bound in Section~\ref{sec:upper-combine}.

\subsection{Regularity Conditions}\label{sec:regcon}

In this section we state a set of three deterministic regularity conditions. We will then show that the set of uncorrupted nodes of a random Erd\H{o}s-R\'enyi graph satisfy these regularity conditions with high probability.
First, we define the following quantities $\Bounda$ and $\Boundb$, which we use in stating the regularity conditions and in the bounds of several lemmas and theorems.
For $p\in [0,1]$ and $n> 0$, let 
\begin{align}
\Boundb(p,n) := c\cdot\max\Bigg(\sqrt{\frac{p(1-p)}{n}} ,\frac{\sqrt{\ln n}}{n}\Bigg).\label{eqn:def-eta}
\end{align}
For $\alpha\in (0,1]$, $p\in [0,1]$ and $n> 0$, let 
\begin{align}
\Bounda(\alpha,p,n) := c_1\cdot\max\Bigg(\alpha\sqrt{\frac{p}{n}\ln\frac{e}{\alpha}} ,\frac{\alpha}{n}\ln\frac{e}{\alpha}, \frac{\sqrt{p\ln n}}{n}\Bigg).\label{eqn:def-kappa}
\end{align}

\noindent In the above definitions $c$ and $c_1$ are some constants that we determine in Theorem~\ref{th:goodspecnorm}.

\medskip 
\noindent  We employ the following regularity conditions.
\begin{definition}
\label{def:regular}
Given $\regpar_1\in [0,1/2)$, $\regpar_2\in[0,1/2)$, and an $[n]\times[n]$ adjacency matrix $A$, a set of nodes $\goodnodes\subseteq[n]$ of the graph corresponding to $A$ satisfy \emph{$(\regpar_1,\regpar_2,p)$-regularity} if
\begin{enumerate}
\item \label{pt:zero-point} $|\goodnodes^c|\le \regpar_1 n$.
    \item \label{pt:first-point} For all $\goodnodesub\subseteq \goodnodes $, 
    \[ \|(A - p)_{\goodnodesub\times \goodnodesub}\| \le n\cdot \Boundb(p,n). \]
   \item   \label{pt:second-point} For all $\goodnodesub , \goodnodesubanother\subseteq \goodnodes $ such that $|\goodnodesub |,|\goodnodesubanother|\in [0, \regpar_2 n]\cup [n-\regpar_2 n, n]$, then
\[
\Bigg|{\sum_{i\in \goodnodesub }\sum_{j\in \goodnodesubanother}(A_{i,j}-p)}\Bigg|\le   n^2\cdot\Bounda(\regpar_2,p,n). 
\]
\end{enumerate}
\end{definition}

Item~\ref{pt:first-point} implies that upon subtracting $p$ from each entry of the adjacency matrix $\AdjM$, the spectral norm of the matrix corresponding to all subgraphs of the subgraph $\goodnodes\times\goodnodes$ is bounded. 
Item~\ref{pt:second-point} implies that upon subtracting $p$ from each entry of the adjacency matrix $\AdjM$, the sum of the entries over any of its submatrices $\goodnodesub\times\goodnodesubanother\subseteq\goodnodes\times \goodnodes$ has a small absolute value,  as long as each of $\goodnodesub$ and $\goodnodesubanother$ either leave out or include at most $\regpar_2 n$ nodes.
We will informally refer to nodes in the set $\goodnodes\subseteq[n]$ that satisfy $(\regpar_1,\regpar_2,p)$-regularity as \emph{good nodes}.


For a subset $S\subseteq[n]$ and adjacency matrix $A$, we will use {$p_S:= \frac{\sum_{i,j \in S} A_{i,j}}{|S|^2}$} to denote (approximately) the empirical fraction of edges present in the subgraph induced by a set $S$.
Note that this differs slightly from expression one might anticipate,
$\binom{|S|}{2}^{-1}\left(\sum_{i < j:\ i,j \in S} A_{i,j}\right)$.
For convenience, our sum double-counts each edge and also includes the $A_{i,i}$ terms (which are always 0 due to the lack of self-loops). 
The double counting is accounted for since the denominator is scaled by a factor of $2$.
The inclusion of the diagonal $0$'s is \emph{not} accounted for, thus leading to $p_S$ being a slight under-estimate of the empirical edge parameter for this subgraph, but not big enough to make a significant difference.

The following lemma lists some simple but useful consequences of the regularity conditions that we use in later proofs.
\begin{lemma}
Suppose $0\le \regpar_1,\regpar_2< 1/2$ and adjacency matrix $\AdjM$ has a node subset  $\goodnodes\subseteq [n] $ that satisfies $(\regpar_1,\regpar_2,p)$-regularity, then 
\begin{enumerate}
 \item For all $\goodnodesub\subseteq \goodnodes $,
 \begin{align}\label{eq:llemc}
     \|(\AdjM - p_{\goodnodesub})_{\goodnodesub\times \goodnodesub}\|\le 2n\cdot \Boundb(p,n).
\end{align}
 \item For all $\goodnodesub\subseteq \goodnodes $ of size $\ge (1-\regpar_2)n$,
\begin{align}\label{eq:llemb}
|p_{\goodnodesub}-p|\le {4 \Bounda(\regpar_2,p,n)}.
\end{align}
\end{enumerate}
\end{lemma}
\begin{proof}
We first prove Equation~\eqref{eq:llemc}. From the triangle inequality
\[
 \|(\AdjM - p_{\goodnodesub})_{\goodnodesub\times \goodnodesub}\| \le  \|(\AdjM - p)_{\goodnodesub\times \goodnodesub}\|+|p - p_{\goodnodesub}|\cdot {\goodnodesub}.
\]
From Lemma~\ref{lem:specnormlb} we have
\[
\|(\AdjM - p)_{\goodnodesub\times \goodnodesub}\|\ge |\goodnodesub|\cdot|p_{\goodnodesub}-p| .
\]
Combining the above two equations with regularity proves Equation~\eqref{eq:llemc},
\[
\|(\AdjM - p_{\goodnodesub})_{\goodnodesub\times \goodnodesub}\| \le 2\|(\AdjM - p)_{\goodnodesub\times \goodnodesub}\|\le 2n\cdot \Boundb(p,n),
\]
where the last inequality follows from regularity condition~\ref{pt:first-point}.

Next, Equation~\eqref{eq:llemb} is obtained by using $\goodnodesub=\goodnodesubanother$ in regularity condition~\ref{pt:second-point} and $|\goodnodesub|\ge n/2$.
\end{proof}

Equation~\eqref{eq:llemc} implies that if the adjacency matrix of any subset of good nodes is centered around its empirical fraction of the edges, then its spectral norm is bounded.
Equation~\eqref{eq:llemb} implies that for any subset of good nodes that excludes at most $\regpar_2n$ nodes, the empirical fraction of edges in the subgraph induced by it estimates $p$ accurately. 

The next theorem shows that the set of uncorrupted nodes of a random Erd\H{o}s-R\'enyi graph satisfy these regularity conditions with high probability.
\begin{theorem}\label{th:goodspecnorm} 
For any $\corr \in [0,1/2)$, $n>0$ and $p>0$, let $A$ be a $\corr$-corrupted adjacency matrix of a sample from $G(n,p)$. 
There exist universal constants $c$ and $c_1$ in Equations~\eqref{eqn:def-eta} and~\eqref{eqn:def-kappa}, respectively,  such that with probability at least $1-4n^{-2}$ the set of uncorrupted nodes $\goodnodes$ satisfy $(\regpar_1,\regpar_2,p)$-regularity for all $\regpar_1\in [\gamma,1/2]$ and $\regpar_2\in [0,1/2]$. 
\end{theorem}

\begin{proof}
In a $\corr$-corrupted graph the set of uncorrupted nodes $\goodnodes$ has size $\ge (1-\corr)n$, which proves regularity condition~\ref{pt:zero-point}.

We use the following bound on the spectral norm of a centered version of $\tilde A$, which follows from Remark 3.13 of~\cite{BandeiraVH16}. 
\begin{lemma}\label{lem:goodspecnormnew}
Let $\tilde A$ be the adjacency matrix of a sample from $G(n,p)$ and $I$ be the $n\times n$ identity matrix. There exist a universal constant $c$ such that with probability at least $1-n^{-2}$,
$\norm{\tilde A - p +p I} \le c\sqrt{np(1-p)+\ln n}$.  
\end{lemma}
To establish regularity condition~\ref{pt:first-point}, note that $A$ and $\tilde A$ agree on $(i,j)\in \goodnodes \times \goodnodes $, and therefore by Lemma~\ref{lem:spsubmat} and Lemma~\ref{lem:goodspecnormnew}, $\|(A - p)_{\goodnodesub \times \goodnodesub }\|\!=\! \|(\tilde A - p)_{\goodnodesub \times \goodnodesub }\| \le \|(\tilde A - p)\|\le \|(\tilde A - p+pI)\|+p\|I\|\le c\sqrt{np(1-p)+\ln n}+1$.

The following theorem implies regularity condition~\ref{pt:second-point}. The proof uses a Chernoff and union bound style argument, and is provided in Section~\ref{sec:app-mainconth}.
\begin{theorem}\label{th:mainconth} 
Let $\tilde A$ be the adjacency matrix of a sample from $G(n,p)$. 
With probability at least $1-3 n^{-2}$, simultaneously for all $\alpha\in[0,\frac{1}{2}]$, we have
\[
\max_{|S|,|S'|\in \setsize_\alpha}\Bigg|{\sum_{i\in S,\,j\in S'}(\tilde A_{i,j}-p)}\Bigg|\le 6\max\Big\{ 16\alpha n\sqrt{{p n}\ln\frac{e}{\alpha }},60\alpha n\ln\frac{e}{\alpha },5n\sqrt{p\ln (en)}\Big\},
\]
where we define $\setsize_{\alpha}:= [0, \alpha n]\cup [n-\alpha n, n]$.
\end{theorem}
\end{proof}

\subsection{An Inefficient Coarse Estimator}\label{sec:conseqreg}
\label{sec:Reg-graphs}

In this section we propose a simple inefficient algorithm to recover a coarse estimate of $p$, which has an optimal dependence on all parameters other than $\regpar_1$.


The following theorem serves as the foundation of our coarse estimator.
It shows that if, for any subset $S\subseteq[n]$ of size $\ge n/2$ nodes, the spectral norm of its submatrix centered with respect to $p_S$ is small, then $p_S$ is a reasonable estimate of $p$.
\begin{theorem}\label{th:newth}
Suppose $0\le \regpar_1,\regpar_2< 1/2$, and let $\AdjM$ be an adjacency matrix containing a $(\regpar_1,\regpar_2,p)$-regular subgraph.
Then for all $S\subseteq [n]$ such that $|S|\ge n/2$, we have 
\[
 |p_{S}-p|\le \frac{\|(\AdjM - p_{S})_{S\times S} \|+n\cdot \Boundb(p,n)}{(1/2-\regpar_1)n }.
\]
\end{theorem}
\begin{proof}
Let $\goodnodes$ be the $(\regpar_1,\regpar_2,p)$-regular subgraph of $\AdjM$.
From the triangle inequality,
\[
\|(\AdjM - p_{S})_{(S\cap \goodnodes)\times (S\cap \goodnodes)} \| \ge |(p- p_{S})\cdot(S\cap \goodnodes)|-\|(\AdjM - p)_{(S\cap \goodnodes)\times (S\cap \goodnodes)} \|.
\]
Then by Lemma~\ref{lem:spsubmat},
\[
|(p- p_{S})\cdot(S\cap \goodnodes)|\le \|(\AdjM - p_{S})_{(S\cap \goodnodes)\times (S\cap \goodnodes)} \| +\|(\AdjM - p)_{(S\cap \goodnodes)\times (S\cap \goodnodes)} \|\le  \|(\AdjM - p_{S})_{S\times S} \| +\|(\AdjM - p)_{ \goodnodes\times \goodnodes} \|.
\]
Finally, noting that $|S\cap \goodnodes|\ge |S|-|\goodnodes^c|\ge |S|-\regpar_1 n\ge n/2-\regpar_1 n$ proves the theorem.
\end{proof}

With this in hand, it suffices to locate a subset of nodes $S$ such that $\|(\AdjM - p_{S})_{S\times S} \|$ is small. We provide the accuracy guarantee of our inefficient algorithm in the following theorem.

\begin{theorem}\label{th:ineff}
Suppose $0\le \regpar_1,\regpar_2< 1/2$, and let $\AdjM$ be an adjacency matrix containing a $(\regpar_1,\regpar_2,p)$-regular subgraph.
Let
\[\hat S = {\arg\min}_{S\subseteq [n]:|S|\ge n/2}\ \|(\AdjM - p_{S})_{S\times S} \|. \]
Then  $\|(\AdjM - p_{\hat S})_{\hat S\times \hat S} \|\le 2n\cdot \Boundb(p,n)$ and $|p_{\hat S}-p|\le \frac{3}{(1/2-\regpar_1) }\cdot \Boundb(p,n)$.
\end{theorem}
\begin{proof}
Let $\goodnodes$ be the $(\regpar_1,\regpar_2,p)$-regular subgraph of $\AdjM$.
From the definition of $\hat S$,
\[\|(\AdjM - p_{\hat S})_{\hat S\times \hat S} \|\le \|(\AdjM - p_{\goodnodes})_{\goodnodes\times \goodnodes} \|\le2n\cdot \Boundb(p,n), 
\]
where the last inequality uses Equation~\eqref{eq:llemc}. The proof follows from Theorem~\ref{th:newth}.
\end{proof}

Theorem~\ref{th:ineff} implies the following simple algorithm to estimate $p$: compute $\hat S$ by iterating over all subsets of $[n]$, and then output $p_{\hat S}$.
Combining with Theorem~\ref{th:goodspecnorm}, this proves Theorem~\ref{thm:large-gamma}. 
The clear downside of this approach is that it is not computationally efficient, with a running time that depends exponentially on $n$.
Also, as we will later establish, while this algorithm gives near-optimal rates for all constant $\corr$ bounded away from $1/2$ by a constant, it may be sub-optimal for smaller $\corr$.
In the following sections, we address both of these issues: we provide a computationally efficient algorithm which provides near-optimal rates for $\corr < 1/60$.

\subsection{An Efficient Coarse Spectral Algorithm}
\label{sec:upper-spectral}
In this section, we propose an efficient spectral method (Algorithm~\ref{algo:part}) which finds a subset $S^*\subseteq[n]$ such that both the set $(S^*)^c$ 
and the spectral norm $\|(A- p_{S^*})_{S^*\times S^*} \|$ are small.
Note that the latter guarantee is comparable to the inefficient algorithm from Section~\ref{sec:conseqreg}.
Then Theorem~\ref{th:newth} implies that $p_{S^*}$ is an accurate estimate of $p$.
We note that this is still a \textit{coarse} estimate of $p$, which has a sub-optimal dependence on $\regpar_1$.\footnote{The guarantees are comparable to Theorem~\ref{th:ineff}, up to constant factors.} 
In the following section, we will post-process the set $S^*$ returned by Algorithm~\ref{algo:part} to provide our near-optimal bounds. 

\begin{theorem}\label{thm:specnorm}
Suppose $\regpar_1\in [\frac{1}{{n}}, \frac{1}{60}]$, $\regpar_2\in[0,1/2]$ and let $\AdjM$ be an adjacency matrix containing an $(\regpar_1,\regpar_2,p)$-regular subgraph.
With  probability at least $1-{1/n^2}$,\footnote{The probability of success of Algorithm~\ref{algo:part} is $\Pr[\text{Bin}(\lfloor9\regpar_1 n\rfloor,0.15)\ge \lfloor\regpar_1 n\rfloor]\ge 1-\exp(-\Omega(\regpar_1 n))$. Note that for all values of $\regpar_1 n$ the success probability is $>1/2$. When $\regpar_1 n= \Omega(\log n)$ then it gives the probability of success at least $1-1/n^2$.  When $\regpar_1 n= \mathcal O(\log n)$, to get the probability of success $\ge 1-1/n^2$ one can run Algorithm~\ref{algo:part} $\mathcal O(\log n)$ times and choose an $S^*$ for which $\|(A - p_{S^*})_{S^*\times S^*}\|$ is the minimum among all runs.} Algorithm~\ref{algo:part} returns a subset $S^*$  with {$|S^*|\ge (1-9\regpar_1)n$} such that $\|(A- p_{S^*})_{S^*\times S^*} \|\le 20n\cdot \Boundb(p,n)$. Furthermore, these conditions on $S^*$ imply $|p_{S^*}-p|\le 45\cdot \Boundb(p,n)$.
\end{theorem}

\begin{algorithm}[H]
  \begin{algorithmic}
    \Require{number of nodes $n$, parameter $\regpar_1\in [1/n, 1/60]$, adjacency matrix $A$}
    \State{$S\gets [n]$, Candidates $\gets\{\}$}
    \State{$\text{Candidates} \gets \text{Candidates} \cup \{S\}$} 
    \For{$t = 1$ to $9\regpar_1 n$}
        \State{Compute a top normalized eigenvector $v$ of the matrix $(A - p_S)_{S\times S}$}
        \State{Draw $i_t$ from the distribution where $i \in S$ is selected with probability $v_i^2$}
        \State{$S \gets S \setminus \{i_t\}$}
        \State{$\text{Candidates} \gets \text{Candidates} \cup \{S\}$}    \EndFor
    \State{$S^*\gets \arg\min_{S\in \text{Candidates}} \|(A - p_S)_{S\times S}\|$}\\
     \Return $S^*$ 
  \end{algorithmic}
  \caption{\label{algo:part}Spectral algorithm for estimating $p$}
\end{algorithm}

In the remainder of this section we will prove that Algorithm~\ref{algo:part} indeed outputs a subset $S^*$ with the guarantee in Theorem~\ref{thm:specnorm}. 
The key technical argument is that if the spectral norm of $(A - p_S)_{S\times S}$ is large, the normalized top eigenvector $v$ of $(A - p_S)_{S\times S}$ places constant weight on the subset $S \cap \goodnodes^c$.
Thus, if at a given iteration Algorithm~\ref{algo:part} possesses an unsatisfactory set $S$, it will remove a node from $S \cap \goodnodes^c$ with a constant probability. We formalize this argument in Lemma~\ref{lem:delprob} and Theorem~\ref{th:lacon}.


\begin{lemma}
\label{lem:delprob}
Suppose $\regpar_1\in [\frac{1}{{n}}, \frac{1}{60}]$, $\regpar_2\in[0,1/2]$ and let $\AdjM$ be an adjacency matrix containing an $(\regpar_1,\regpar_2,p)$-regular subgraph.
Let $S \subseteq [n]$ be of size $|S| \geq (1 - 9\regpar_1)n$, and $v$ be the normalized top eigenvector of $(A-p_S)_{S\times S}$.
If $\|(A-p_S)_{S\times S}\|\ge 20n\cdot \Boundb(p,n)  $ then $\|v_{S\cap \goodnodes^c}\|^2\ge 0.15$.
\end{lemma}
\begin{proof}
We first require the following lemma, which lower bounds the spectral norm of a matrix $(A-p_S)_{S\times S}$ primarily in terms of the empirical estimates of $p$ corresponding to the submatrices induced by $S$ and $S \cap \goodnodes$.
The proof appears in Section~\ref{sec:whatisinthename}.

\begin{lemma}\label{lem:whatisinthename}
Given any symmetric matrix $A$, and subsets $S,\goodnodes\subseteq[n]$
\[
\|(A-p_S)_{S\times S}\|\ge  \frac{|p_{S\cap \goodnodes }-p_S|\cdot|S\cap \goodnodes |}{3}\cdot \min\left\{\sqrt{\frac{|S\cap \goodnodes |}{|S\cap \goodnodes^c|}}, \frac{|S\cap \goodnodes |}{|S\cap \goodnodes^c|}\right\}.
\]
\end{lemma}

For $\regpar_1\le 1/60$ and $|S|\ge (1-9\regpar_1)n$, we can deduce that $|S\cap \goodnodes|\ge n(1-10\regpar_1)\ge 5n/6$ and $|S\cap \goodnodes^c|\le |\goodnodes^c|\le\regpar_1 n\le n/60$. 
Therefore, ${|S\cap \goodnodes^c|}/{|S\cap \goodnodes|}\le 1/50.$
By Lemma~\ref{lem:whatisinthename},
\[
|S\cap \goodnodes|\cdot|p_{S\cap \goodnodes}-p_S|\le 3\|(A-p_S)_{S\times S}\|\max\Big\{\sqrt{\frac{|S\cap \goodnodes^c|}{|S\cap \goodnodes|}}, \frac{|S\cap \goodnodes^c|}{|S\cap \goodnodes|}\Big\}\le \frac{3}{\sqrt{50}}\cdot \|(A-p_S)_{S\times S}\|.
\]
Applying Equation~\eqref{eq:llemc} with $\goodnodes'= S\cap \goodnodes$, we have
\[
\|(A-p_{S\cap\goodnodes})_{(S\cap \goodnodes)\times(S\cap \goodnodes)}\|\le  2n\cdot \Boundb(p,n).
\]
This implies $\|(A-p_{S\cap\goodnodes})_{(S\cap \goodnodes)\times(S\cap \goodnodes)}\|\le 0.1 \|(A-p_S)_{S\times S}\|$. Next, by the triangle inequality, 
\begin{align*}
\|(A-p_S)_{(S\cap \goodnodes)\times(S\cap \goodnodes)}\|
&\le  \|(A-p_{S\cap\goodnodes})_{(S\cap \goodnodes)\times(S\cap \goodnodes)}\|+|S\cap \goodnodes|\cdot|p_{S\cap \goodnodes}-p_S|\\
&\le \bigg(\frac{1}{10}+\frac{3}{\sqrt{50}}\bigg)\|(A-p_S)_{S\times S}\|.
\end{align*}

To interpret the derivation above: we have reasoned that if the spectral norm of $(A-p_S)_{S\times S}$ is large, the contribution due to $S \cap \goodnodes$ (i.e., the submatrix induced by the intersection with the good nodes) is {relatively small.} 
{This suggests that any top eigenvector must place a constant mass on $S \cap \goodnodes^c$.
Indeed, the following theorem formalizes this reasoning, showing that the normalized top eigenvector contains significant weight in this complementary subset of indices.
The proof appears in Section~\ref{sec:lacon}.}

\begin{theorem}\label{th:lacon}
Let $M$ be a non-zero $n\times n$ real symmetric matrix such that for some set $S\subseteq [n]$ and $0\leq \rho\leq1$ we have $\|M_{S\times S}\|\le \rho\|M\|$. 
Let $v$ be any normalized top eigenvector of $M$.
Then
$\|v_{S^c}\|^2\ge {\frac{(1-\rho)^2}{1+(1-\rho)^2}}$.
\end{theorem}

Applying Theorem~\ref{th:lacon} with $\rho=\frac{1}{10}+\frac{3}{\sqrt{50}}$ implies that $\|v_{S\setminus (S\cap \goodnodes)}\|^2=\|v_{S\cap \goodnodes^c}\|^2\ge \frac{(1-\rho)^2}{1+(1-\rho)^2}> 0.15$.
\end{proof}

We conclude this section with the proof of Theorem~\ref{thm:specnorm}. 

\begin{proof}[Proof of Theorem~\ref{thm:specnorm}]
It suffices to show that at least one of the sets $S$ encountered by Algorithm~\ref{algo:part} satisfies  the condition  $\|(A - p_S)_{S\times S}\|\le 20n\cdot \Boundb(p,n)$.
From Lemma~\ref{lem:delprob} it follows that until the algorithm finds such a subset $S$, the probability of deleting a good node in each deletion step is at least $0.15$. 
Since there are $9 \regpar_1 n$ steps, a standard Chernoff-style argument implies that either a subset $S$ (including nodes from both $\goodnodes^c$ and $\goodnodes$) satisfying the conditions of the theorem will be created, or with probability at least $\Pr[\text{Bin}(\lfloor9\regpar_1 n\rfloor,0.15)\ge \lfloor\regpar_1 n\rfloor]\ge 1-\exp(-\Omega(\regpar_1 n))$, all nodes from $\goodnodes^c$ will be deleted and thus $S \subseteq \goodnodes$.
In the latter case we apply Equation~\eqref{eq:llemc}, which implies that $\|(A - p_S)_{S\times S}\|\le 20n\cdot \Boundb(p,n) $ and the theorem.
\end{proof}


\begin{remark} 
\label{rmk:running-time}
{Algorithm~\ref{algo:part} runs for $9\regpar_1  n$ rounds, and in each round the algorithm finds the top eigenvector of an $n\times n$ matrix. 
This may be expensive to compute when the spectral gap is small.
However, Lemma~\ref{lem:comp} shows it suffices to find any unit vector $v\in \mathbb R^n$ such that $|v^\intercal (A - p_S)_{S\times S} v|\ge 0.99 ||(A - p_S)_{S\times S}||$. 
Note that such a unit vector can be found in $\tilde O(n^2)$ time~\cite{MuscoM15}. 
Therefore, one can implement Algorithm~\ref{algo:part} to run in  
$\tilde O(\regpar_1 n^3)$ time.}
\end{remark}

\subsection{A Fine Trimming Algorithm}
\label{sec:upper-trimming}
In this section, we provide a trimming method (Algorithm~\ref{algo:trim}), which refines the result of Algorithm~\ref{algo:part}, improving its guarantee (quantified in Theorem~\ref{thm:specnorm}) by {up to} a factor of $\regpar_1$.

\begin{theorem}
\label{thm:main-bound}
Let $\regpar_1\in [\frac{1}{{n}}, \frac{1}{60}]$, and $\AdjM$ be an adjacency matrix containing an $(\regpar_1,13\regpar_1,p)$-regular subgraph. Suppose we have some $S^*$ such that {$|S^*|\ge (1-9\regpar_1)n$ and} $\|(A - p_{S^*})_{S^*\times S^*}\|\le 20n\cdot \Boundb(p,n)$,
Algorithm~\ref{algo:trim} outputs {$p_{S^f}$} such that for some universal constants $c_2,c_3>0$,
\[
\big|p_{S^f}-p\big|\le   c_2\regpar_1 \Boundb(p,n)+ c_3\Bounda(13\regpar_1,p,n).
\]
\end{theorem}

The algorithm is easy to describe. 
For a subset $S^*\subseteq[n]$ and a node $i\in S^*$, we define $p_{S^*}^{(i)} := \frac{\sum_{j\in {S^*}}A_{i,j}}{|{S^*}|}$ to be the normalized degree of node $i$ in the subgraph induced by $S^*$. 
We remove the $3\regpar_1 n$ nodes whose normalized degree on subgraph $S^*$ deviate furthest from the average parameter $p_{S^*}$

\begin{algorithm}[H]
  \begin{algorithmic}
    \Require{number of nodes $n$, parameter $\regpar_1\in [1/{n}, 1/60]$, adjacency matrix $A$, subset $S^* \subseteq [n]$}
    \State{Define the score for each node $i\in S^*$ to be $|p_{S^*}-p_{S^*}^{(i)}|$}
    \State{Remove the $3\regpar_1 n$ nodes in $S^*$ with the highest scores to obtain $S^f$}\\
     \Return {$p_{S^f} $} 
  \end{algorithmic}
  \caption{\label{algo:trim}Trimming Algorithm}
\end{algorithm}


The first lemma shows that the average of entries of all \textit{small submatrices} of $S^*\times S^*$ are close to $p_{S^*}$.
\begin{lemma}\label{lem:somecor}
Assume the conditions of Theorem~\ref{thm:main-bound} hold.
For all $S_1,S_2\subseteq {S^*}$ with $|S_1|,|S_2|\le 3\regpar_1 n$ we have
\begin{align*}
\Bigg|\sum_{i\in S_1,\, j\in S_2}(A_{i,j}-p_{S^*})\Bigg|\le 60\regpar_1n^2\cdot \Boundb(p,n).
\end{align*}
\end{lemma}
\begin{proof} Since $\|(A - p_{S^*})_{S^*\times S^*}\|\le 20n\cdot \Boundb(p,n) $, and using Lemma~\ref{lem:spsubmat} and Lemma~\ref{lem:specnormlb}, we get
\[
\Bigg|\sum_{(i,j)\in S_1\times S_2}(A_{i,j}-p_{S^*})\Bigg|\!\le\! \sqrt{|S_1|\cdot |S_2|}\cdot \|(A-p_{S^*})_{S_1\times S_2}\| \!\le 3\regpar_1 n\|(A-p_{S^*})_{{S^*}\times {S^*}}\|\!\le \! 60\regpar_1n^2\cdot \Boundb(p,n).
\]\end{proof}

We now show that all the nodes in $S^{f}$ have normalized degree close to $p_{S^*}$. 
\begin{lemma}\label{lem:trimnodes}
Assume the conditions of Theorem~\ref{thm:main-bound} hold, and let $S^f$ be the output of Algorithm~\ref{algo:trim}, then for every node $i\in S^f$,
\[
|p_{S^*}^{(i)}-p_{S^*}|\le \Big(\frac{2\Bounda(13\regpar_1,p,n)}{\regpar_1}+ 210\Boundb(p,n)\Big).
\]
\end{lemma}
\begin{proof}
Suppose to the contrary that after $3\regpar_1 n$ nodes are deleted by Algorithm~\ref{algo:trim}, there is a node $i\in S^f$ such that $|p_{S^*}^{(i)}-p_{S^*}|> \Big(\frac{2\Bounda(13\regpar_1,p,n)}{\regpar_1}+ 210\Boundb(p,n)\Big)$. Therefore, all the nodes deleted by Algorithm~\ref{algo:trim} are such that $|p_{S^*}^{(i)}-p_{S^*}|> \Big(\frac{2\Bounda(13\regpar_1,p,n)}{\regpar_1}+ 210\Boundb(p,n)\Big)$. Let $D^+$ be the set of nodes deleted by Algorithm~\ref{algo:trim} such that $p_{S^*}^{(i)}>p_{S^*}$ for $i\in D^+$ 
and  $D^-$ be the set of deleted nodes $i$ such that $p_{S^*}^{(i)}<p_{S^*}$ for $i\in D^-$. Since $|D^+|+|D^{-}|=3\regpar_1 n$ {and $|(D^+ \cup D^-)\setminus\goodnodes|\le |\goodnodes^c|\le\regpar_1 n$,} we have that $|D^+\cap \goodnodes|\ge \regpar_1 n$ or $|D^-\cap \goodnodes|\ge \regpar_1 n$. Suppose $|D^+\cap \goodnodes|\ge \regpar_1 n$. Let $\goodnodesub=D^+\cap \goodnodes$. Then, using $|\goodnodesub|\ge\regpar_1 n$ and $|S^*|>n/2$, we have
\begin{align*}
\sum_{i\in \goodnodesub, j\in S^*} (A_{i,j}-p_{S^*})= \sum_{i\in \goodnodesub} |S^*| (p_{S^*}^{(i)}-p_{S^*}) &> |\goodnodesub|\cdot|S^*|\cdot\Big(\frac{2\Bounda(13\regpar_1,p,n)}{\regpar_1}+ 210\Boundb(p,n)\Big)\\
&\ge n^2\Bounda(13\regpar_1,p,n)+ 105|\goodnodesub|n\cdot \Boundb(p,n).
\end{align*}
Now, note that 
\begin{align*}
\sum_{i\in \goodnodesub, j\in S^*} (A_{i,j}-p_{S^*}) = \sum_{i\in \goodnodesub, j\in S^*\cap \goodnodes} (A_{i,j}-{p})+ |\goodnodesub|\cdot |S^*\cap \goodnodes| \cdot (p-p_{S^*}) + \sum_{i\in \goodnodesub, j\in S^*\cap \goodnodes^c} (A_{i,j}-p_{S^*})
\end{align*}
By Lemma~\ref{lem:somecor} with $S_1=\goodnodesub$ and $S_2=S^*\cap \goodnodes^c$ the last term in the expression above is at most $60~\regpar_1~{n^2\cdot}\Boundb(p,n)$. 
For the second term note that $|p-p_{S^*}|<  45\cdot \Boundb(p,n)$ and therefore, 
the second term is at most $45|\goodnodesub|n\cdot \Boundb(p,n)$. Finally using regularity condition~\ref{pt:second-point} with $\goodnodesub$ and $\goodnodesubanother=S^*\cap \goodnodes$ and $\regpar_2=13\regpar_1$ bounds the first term by $n^2\cdot\Bounda(13\regpar_1,p,n)$. 
Combining the three bounds and using $|\goodnodesub|\ge \alpha_1 n$, 
\[
\sum_{i\in \goodnodesub, j\in S^*} (A_{i,j}-p_{S^*})\le n^2\Bounda(13\regpar_1,p,n)+ (45|\goodnodesub|+ 60\regpar_1{n})n\cdot \Boundb(p,n)\le n^2\Bounda(13\regpar_1,p,n)+ 105|\goodnodesub|n\cdot \Boundb(p,n),\] 
This shows the contradiction and completes the proof for the case $|D^+\cap \goodnodes|>\regpar_1 n$.
The case when $|D^-\cap \goodnodes|>\regpar_1 n$ has a similar argument and is omitted.
\end{proof}

Combining these lemmas appropriately allows us to conclude our main result on the guarantees of Algorithm~\ref{algo:trim}.
\begin{proof}[Proof of Theorem~\ref{thm:main-bound}] 
We will partition $S^f\times S^f$ into the following groups and bound each term separately. 
\begin{align*}
\sum_{i,j\in S^f}A_{i,j} = \sum_{i,j\in S^f\cap \goodnodes} A_{i,j}+  2\sum_{i\in S^f,j\in S^f\cap \goodnodes^c} A_{i,j}-  \sum_{i,j\in S^f\cap \goodnodes^c} A_{i,j}. 
\end{align*}
Since $p_{S^f}=\sum_{i,j\in S^f}A_{i,j}/{|S^f|^2}$, by the triangle inequality, 
\begin{align*}
\big|{p_{S^f}}-p\big|
&\le
\Bigg|\frac{\sum_{i,j\in S^f\cap \goodnodes}(A_{i,j}-p)}{|S^f|^2}\Bigg|
+2\Bigg|\frac{\sum_{i\in S^f  \,j\in S^f\cap \goodnodes^c}(A_{i,j}-p)}{|S^f|^2}\Bigg|
+\Bigg|\frac{\sum_{i,j\in S^f\cap \goodnodes^c}(A_{i,j}-p)}{|S^f|^2}\Bigg|.
\end{align*}

For the first term, $|S^f\cap \goodnodes|\ge  (1-13\regpar_1)n\ge n/2$. Using Equation~\eqref{eq:llemb} with $\goodnodesub=S^f\cap \goodnodes$ and $\regpar_2=13\regpar_1$,
\[
{\Bigg|\frac{\sum_{i,j\in S^f\cap \goodnodes}(A_{i,j}-p)}{|S^f|^2}\Bigg|\le}\Bigg|\frac{\sum_{i,j\in S^f\cap \goodnodes}(A_{i,j}-p)}{|S^f\cap \goodnodes|^2}\Bigg| = |p_{S^f\cap \goodnodes}-p|\le 4 \Bounda(13\regpar_1,p,n).
\]

Since $|S^f\times (S^f\cap \goodnodes^c)|\le \regpar_1 n^2$, and $|S^f|\ge n/2 $, by the triangle inequality, the second term is bounded by
\begin{align}
2\Bigg|\frac{\sum_{i\in S^f  \,j\in S^f\cap \goodnodes^c}(A_{i,j}-p)}{|S^f|^2}\Bigg|
&\le 2\Bigg|\frac{\sum_{i\in S^f  \,j\in S^f\cap \goodnodes^c}(A_{i,j}-p_{S^*})}{n^2/4}\Bigg| + \frac{2\regpar_1n^2}{n^2/4} |p_{S^*}-p|\nonumber\\
&\le 8\Bigg|\frac{\sum_{i\in S^*  \,j\in S^f\cap \goodnodes^c}(A_{i,j}-p_{S^*})}{n^2}\Bigg|+8\Bigg|\frac{\sum_{i\in S^*\setminus S^f  \,j\in S^f\cap \goodnodes^c}(A_{i,j}-p_{S^*})}{n^2}\Bigg| \nonumber\\
&+8\regpar_1\cdot45\cdot \Boundb(p,n)\nonumber.
\end{align}

Since $|S^*\setminus S^f|\le 3\regpar_1 n$ and $|S^f\cap \goodnodes^c|\le \regpar_1 n$, by taking $S_1=S^*\setminus S^f$ and $S_2 = S^f\cap \goodnodes^c$ in Lemma~\ref{lem:somecor} bounds the second term above by $8(60\regpar_1\cdot \Boundb(p,n))$. For the first term, 
\begin{align*}
\Bigg|\frac{\sum_{i\in S^*  \,j\in S^f\cap \goodnodes^c}(A_{i,j}-p_{S^*})}{n^2}\Bigg|
&\le  \sum_{j\in S^f\cap \goodnodes^c}\Bigg|\frac{\sum_{i\in S^*}(A_{i,j}-p_{S^*})}{n^2}\Bigg| \\
&\le \sum_{j\in S^f\cap \goodnodes^c}\frac{|S^*|}{n^2}\Bigg|\frac{\sum_{i\in S^*}(A_{i,j}-p_{S^*})}{|S^*|}\Bigg|\\
&\le \sum_{j\in S^f\cap \goodnodes^c}\frac{1}{n}|p^{(j)}_{S^*}-p_{S^*}|\\
&\le \regpar_1\cdot\Big(\frac{2\Bounda(13\regpar_1,p,n)}{\regpar_1}+ 210\Boundb(p,n)\Big),
\end{align*}
where we use Lemma~\ref{lem:trimnodes} and $|S^{f}\cap\goodnodes^c|\le\regpar_1 n$. 

For the final term, since  $|(S^f\cap\goodnodes^c)\times (S^f\cap\goodnodes^c)|\le \regpar_1^2n^2$,
\begin{align*}
\Bigg|\frac{\sum_{i,j\in S^f\cap\goodnodes^c}(A_{i,j}-p)}{|S^f|^2}\Bigg|
&\le \Bigg|\frac{\sum_{i,j\in S^f\cap\goodnodes^c}(A_{i,j}-p_{S^*})}{|S^f|^2}\Bigg|+ |p_{S^*}-p|\cdot \frac{|S^f\cap\goodnodes^c|^2}{|S^f|^2},
\end{align*}
which can be bounded again by taking $S_1=S_2=S^f\cap\goodnodes^c$ in Lemma~\ref{lem:somecor}. 
\end{proof}


\subsection{Putting Things Together}
\label{sec:upper-combine}
We now combine our methods from previous sections to prove our main result.
This primarily consists of running Algorithm~\ref{algo:part} followed by Algorithm~\ref{algo:trim}, as described by Algorithm~\ref{algo:complete} and quantified by Theorem~\ref{thm:main-upper-bound-complete}.
For technical reasons, to get the correct scaling of the error with respect to the parameter $p$, we run this procedure on both the graph and its complement, and output the appropriate of the two estimates.
This is described in Algorithm~\ref{algo:erd-ren}, and quantified in Theorem~\ref{thm:upper-complete}.

\begin{theorem}
\label{thm:main-upper-bound-complete}
Suppose $\regpar_1\in [\frac{1}{{n}}, \frac{1}{60}]$  and let $\AdjM$ be an adjacency matrix containing an $(\regpar_1,13\regpar_2,p)$-regular subgraph.
With probability at least $1 - n^{-2}$, Algorithm~\ref{algo:complete} outputs {$p_{S^f}$} such that for some universal constants $c_2,c_3>0$,
\[
\big|p_{S^f}-p\big|\le   c_2\regpar_1 \Boundb(p,n)+ c_3\Bounda(13\regpar_1,p,n).
\]
The running time of this algorithm is $\tilde{O}(\regpar_1 n^3)$.
\end{theorem}
\begin{proof}
The estimation guarantees in Theorem~\ref{thm:main-upper-bound-complete} follows by combining the guarantees of Theorems~\ref{thm:specnorm}, and \ref{thm:main-bound}.
We conclude the proof by analyzing the running time.
As discussed in Remark~\ref{rmk:running-time}, Algorithm~\ref{algo:part} can be implemented in $\tilde {\mathcal O}(\regpar_1n^3)$ time. 
Algorithm~\ref{algo:trim} takes ${\mathcal O}(n^2)$ time.
Hence, Algorithm~\ref{algo:complete} runs in $\tilde {\mathcal O}(\regpar_1 n^3)$ time.
\end{proof}

\begin{algorithm}[H]
  \begin{algorithmic}
    \Require{number of nodes $n$, parameter $\regpar_1\in [1/n, 1/60]$, adjacency matrix $A$}
    \State{$S^* \gets$ run the spectral algorithm (Algorithm~\ref{algo:part}) with inputs $n$, $\regpar_1$, $A$}
    \State{$p_{S^f} \gets$ run the trimming algorithm (Algorithm~\ref{algo:trim}) with inputs $n$, $\regpar_1$,  $A$, $S^*$}
    \State{\Return{$p_{S^f}$}}
  \end{algorithmic}
  \caption{\label{algo:complete} Algorithm for{ estimating $p$}} 
\end{algorithm}

Observe that the $\Bounda(13\regpar_1,p,n)$ error term in Theorem~\ref{thm:main-upper-bound-complete} scales proportional to $\sqrt{p}$, which gives improved error when $p$ is close to $0$.
To enjoy the same improvement for $p$ close to $1$, we can run the algorithm on the complement of the graph.
Theorem~\ref{thm:upper-complete} describes the resulting guarantees, and the procedure appears as Algorithm~\ref{algo:erd-ren}.
{Note that we apply Theorem~\ref{th:goodspecnorm} to convert from adjacency matrices containing regular subgraphs (which we have considered up to this point) back to our original problem.}

\begin{theorem}
\label{thm:upper-complete}
Suppose $\corr\in [\frac{1}{{n}}, \frac{1}{60}]$ and $p \in [0,1]$.
Let $G\sim G(n,p)$, and $A$ be the adjacency matrix of a  rewiring of $G$ by a $\corr$-omniscient adversary. 
With probability at least $1 - 10n^{-2}$, running Algorithm~\ref{algo:erd-ren} will output a $\hat p$ such that
 \begin{align*}
|\hat p - p| \le C\cdot\Paren{ \frac{\sqrt{p(1-p)\log n}}{n}+\frac{\corr\sqrt{p(1-p)\log(1/\gamma)}}{\sqrt{n}}+  \frac\corr n\log n},
 \end{align*}
 for some {universal} constant $C$.
 The running time of this algorithm is $\tilde{O}(\gamma n^3)$. 
 \end{theorem}
\begin{proof}
Theorem~\ref{thm:main-upper-bound-complete} and Theorem~\ref{th:goodspecnorm} imply that with probability $\ge 1-5n^{-2}$, $p^*$ in Algorithm~\ref{algo:erd-ren} satisfies:
\begin{align}\label{eqn:locf}
\big|p^*-p\big|\le   c_2\gamma \Boundb(p,n)+ c_3\Bounda(13\gamma,p,n).   
\end{align}
By symmetry, with probability $\ge 1-5n^{-2}$, $q^*$ in Algorithm~\ref{algo:erd-ren} satisfies:
\begin{align}\label{eqn:locfb}
\big|q^*-(1-p)\big|\le   c_2\gamma \Boundb(1-p,n)+ c_3\Bounda(13\gamma,1-p,n).    
\end{align}
When $p\le 0.1$, equation~\eqref{eqn:locf}  implies $p^*\le 0.5$, and hence $\hat p= p^*$ and $|\hat p -p |=|p^*-p| $. Similarly, when $p\ge 0.9$, \eqref{eqn:locf} implies $p^*> 0.5$, and hence $\hat p= 1- q^*$ and $|\hat p-p|=|(1-q^*)-p|=|(1-p)-q^*|$.
Finally, for $0.1\le p\le 0.9$, we have $|\hat p-p|\le \max\{|p^*-p|,|q^*-(1-p)|\}$. Combining the bound for the three cases completes the proof.
\end{proof}

\begin{algorithm}[H]
  \begin{algorithmic}
    \Require{number of nodes $n$, parameter $\gamma\in [1/n, 1/60]$, adjacency matrix $A$}
    \State{$p^* \gets$ run Algorithm~\ref{algo:complete} with inputs $n$, $\gamma$, $A$}
    \State{$q^* \gets$ run Algorithm~\ref{algo:complete} with inputs $n$, $\gamma$, $(1-I-A)$ ($1$ and $I$ are the $n\times n$ all-ones and identity matrix)}
     \If{$p^* \leq 0.5$}
     \State{$\hat p \gets p^* $}
     \Else{}
    \State{$\hat p \gets 1-q^*$ }
    \EndIf{}
    \State{\Return{$\hat p$}}
  \end{algorithmic}
  \caption{\label{algo:erd-ren} Algorithm for Robust Erd\H{o}s-R\'enyi parameter estimation}
\end{algorithm}

\section{Lower Bounds for Parameter Estimation}
\label{sec:lower-bound}

\def \derg {DG}
In this section, we prove our main lower bound for robust parameter estimation in Erd\H{o}s-R\'enyi random graphs.
This implies that our algorithms are tight up to logarithmic factors.

\begin{theorem}
\label{thm:lower-bound}
Let $p\leq0.5$. Then there exists a $\corr$-oblivious adversary such that no algorithm can distinguish between $G(n,p)$ and $ G\Paren{n,p+0.1\max\Paren{\corr\sqrt{p/n},  \gamma/\ns,  \sqrt{p}/n}}$ with probability at least 0.65.
\end{theorem}
By symmetry, a similar statement holds for $p > 0.5$, with $p$ replaced by $1-p$. 
Combining these two statements gives the lower bound in Theorem~\ref{thm:lower-bound-main}.

To derive our lower bounds, we will consider a directed variant of the Erd\H{o}s-R\'enyi random graph model.

\begin{definition}[Directed Erd\H{o}s-R\'enyi graphs]
 \label{def:bi-gnp}
 The directed Erd\H{o}s-R\'enyi random graph model on $n$ nodes with parameter $p$, denoted as  $\derg (\ns, p)$, is the distribution over directed graphs on $n$ nodes where each edge is present with probability $p$, independently of the other edges. 
\end{definition}

We will show a reduction from $\gamma$-corrupted parameter estimation in the directed Erd\H{o}s-R\'enyi model to the standard one, indicating that the latter is no easier.
We will then prove a lower bound for this setting, which will imply lower bounds for the standard Erd\H{o}s-R\'enyi model.

\begin{lemma}
\label{lem:undirected_to_directed}
Suppose there exists an algorithm for estimating $p$ in $G(n,p)$ under $\corr$-corruptions.
Then there exists an algorithm for estimating $p$ in $\derg(n, p)$ under $\corr$-corruptions with less or equal error.
\end{lemma}
\begin{proof}
We prove this lemma by converting a $\corr$-corrupted graph from $\derg(n,p)$ to a $\corr$-corrupted graph from $G(n,p)$.
Then one can run the algorithm for the undirected setting to obtain an estimate of $p$, which implies the same error guarantees for the directed instance.

Suppose there exists a random directed graph $\derg\sim \derg(n,p)$ which is $\corr$-corrupted by an adversary.
Assume there exists some lexicographic ordering of the nodes (e.g., they are numbered from $1$ to $n$).
We define a corresponding undirected graph $G$ as follows: let there be an edge between nodes $i$ and $j$ in $G$ if there exists an edge from $i$ to $j$ in $DG$ and $i<j$. Sans corruptions, this converts $\derg(n,p)$ into $G(n,p)$ since the edges are still independent and the probability of each edge existing is $p$. Furthermore, when at most $\corr n$ nodes in the original directed graph are modified, at most $\corr n$ nodes are changed in the corresponding undirected graph. 
\end{proof}

For the remainder of this section, we prove an analogue of Theorem~\ref{thm:lower-bound} for directed Erd\H{o}s-R\'enyi graphs.
A subsequent application of Lemma~\ref{lem:undirected_to_directed} will imply Theorem~\ref{thm:lower-bound}.

We start by designing a $\gamma$-oblivious adversary for the directed graph model to prove our lower bound.
For simplicity of the proof, we assume the adversary picks a set $B$ of size $\Bin(n,0.15\gamma)$ to corrupt, by independently picking each node in $[n]$ with probability $0.15\gamma$.
Note that it is possible that the size of the set of corrupted nodes $B$ may exceed $\gamma n$ with small probability. We will address this issue later.

The adversary will only need to corrupt outgoing edges of the nodes in $B$.
More precisely, they will only determine the presence of directed edges $(i,j)$ where $i\in B$ and $j\in [n]$.
The adversary's strategy to corrupt the neighborhood of node $i \in B$ is as follows.
They first choose node $i$'s new out-degree $\deg{(i)}$ independently from some distribution $P$ over $\{0, \dots, n-1\}$.
Then, they select an independent random subset $S_i$ of nodes $[n]\setminus \{i\}$ of size $\deg{(i)}$. Finally, they introduce the directed edge $(i,j)$ for each $j\in S_i$, and remove the directed edge $(i,j)$ for each $j\notin S_i$.
The distribution of the degree of corrupted nodes, $P$, depends on the  parameter $p$ of the Erd\H{o}s-R\'enyi graph and will be specified later.

Given such an adversary, note that the out-degrees $d_1,d_2,\dots,d_n$ are sufficient statistics for estimating $p$.
Observe that the out-degree of any uncorrupted node is distributed as $\Bin(n,p)$, and the out-degree of any corrupted node has distribution $P$. Since each node is corrupted with probability $0.15\gamma$, the out-degree of each node is an i.i.d.\ sample from the mixture distribution $(1-0.15\gamma)\cdot\Bin(n-1,p)+0.15{\gamma}\cdot P$.

Next, we show that for any $p_1\le 1/2$ and $p_2=p_1+0.1\max\Paren{\corr\sqrt{p/n},0.1\corr/n}$ there exist distributions $P_1$ and $P_2$ such that 
\begin{align}\label{eq:labeq}
 (1-0.15\gamma)\cdot\Bin(n-1,p_1)+0.15\gamma \cdot P_1= (1-0.15\gamma)\cdot\Bin(n-1,p_2)+0.15\gamma\cdot P_2.
\end{align}
This will imply that, with the aforementioned adversary, any estimator that distinguishes between the two cases will be correct with probability at most $1/2$.
At this point, we account for the probability that the adversary selects a set $B$ of size $> \corr n$, which is not allowed according to the corruption model. 
By Markov's inequality, this occurs with probability at most $0.15$.
Therefore, even counting such violations as a success at distinguishing the two cases, it still succeeds with probability at most $0.5 + 0.15 = 0.65$.


To prove the existence of $P_1$ and $P_2$ satisfying \eqref{eq:labeq} we use the following folklore fact:
given any two distributions $D_1$ and $D_2$ and $\epsilon>0$, if $d_{\rm TV}\Paren{D_1,D_2}\le \epsilon$, then there exist distributions $Q_1$ and $Q_2$ such that $(1-\epsilon)D_1+\epsilon Q_1 = (1-\epsilon)D_2+\epsilon Q_2$.

Hence, it suffices to show that
\begin{equation}\label{eq:tveqbou}
d_{\rm TV}\Paren{ \Bin(n-1,p_1), \Bin(n-1,p_2)}\le 0.15\gamma.    
\end{equation}
The total variation distance between two binomials can be bounded as~\cite{Roos03a},~\cite[Eq (2.16)]{AdellJ06}.
\begin{align}\label{eq:tv}
d_{\rm TV}\left(\Bin (n',p),\Bin (n',p+x) \right)
&\le \sqrt{\frac e2} \frac{\tau(x)}{(1-\tau(x))^2},
\end{align}
where 
$\tau(x)= x\cdot \sqrt{\frac{n'+2}{2p(1-p)}}$. 
We also use the trivial upper bound $d_{\rm TV}\left(\Bin (n',p),\Bin (n',p+x) \right)\le n'x$.

For the case when $\corr\sqrt{p/n}\ge 0.1\corr/n$, applying the first bound for $x=0.1\corr\sqrt{p/n}$ and $n'=n-1$ we get
\begin{align*}
\tau(x) = 0.1\corr\sqrt{\frac pn}\sqrt{\frac{n+1}{2p(1-p)}}
 \le 0.1 \cdot 1.1 \corr = 0.11\corr.
\end{align*}
For this case, using~\eqref{eq:tv} gives
\begin{align*}
d_{\rm TV}\Paren{ (\Bin(n-1,p_1), \Bin(n-1,p_2)} \le 0.15\corr. 
\end{align*}
For the other case when $\corr\sqrt{p/n}< 0.1\corr/n$, applying the trivial bound gives 
\begin{align*}
d_{\rm TV}\Paren{ (\Bin(n-1,p_1), \Bin(n-1,p_2)}\le 0.1\corr(n-1)/n< 0.1\gamma. 
\end{align*}
This proves~\eqref{eq:tveqbou} and shows the existence of $P_1$ and $P_2$, which completes the proof of the first two terms in the max in Theorem~\ref{thm:lower-bound}.


Finally, we show that the third term in the max in Theorem~\ref{thm:lower-bound} holds even when there is no corruption.
To show this we first note that in absence of corruption the sufficient statistics for estimating $p$ is the total number of edges in the directed graph, which has a distribution $\Bin((n-1)^2,p)$.
Then to show that for $p\le 0.5$ no algorithm can distinguish between between $\derg(\ns, p)$ and $\derg(\ns, p+0.1\sqrt{p}/n)$ with probability $\ge 0.6$ it suffices to show that
$ d_{\rm TV}(\Bin((n-1)^2,p), \Bin((n-1)^2,p+0.1\sqrt{p}/n))<0.2$, which can be verified using~\eqref{eq:tv} for $x= 0.1\sqrt{p}/n$, $n'=(n-1)^2$, \tcr{and any $p<1/2$}.


\bibliographystyle{alpha}
\bibliography{biblio}

\appendix
\section{Concentration Inequalities}
\begin{lemma}[Chernoff bound]
\label{lem:chernoff}
Let $X_1,X_2,...,X_t\sim \text{Ber}\,(p)$ be $t$ independent Bernoulli random variables. Then for any $\lambda>0$
\begin{align}
\Pr\left[\big|\sum^t_{i=1} X_i-tp\big|\ge \lambda \right]\le 2\exp\Paren{-\min\Paren{\frac{\lambda^2}{3tp},\frac{\lambda}{3}}}.
\label{eqn:chernoff}    
\end{align}
\end{lemma}
\ignore{
\begin{proof}
\tcr{
Let $Y_i =1-X_i$.
Then
\[
\Pr\left[\big|\sum^t_{i=1} X_i-tp\big|\ge \lambda \right] = \Pr\left[\big|\sum^t_{i=1} Y_i-t(1-p)\big|\ge \lambda \right].
\]
From the Chernoff bound the term on the right is upper bounded by $2\exp\Paren{-\min\Paren{\frac{\lambda^2}{3tp},\frac{\lambda}{3}}}$ and term on the left is bounded by $2\exp\Paren{-\min\Paren{\frac{\lambda^2}{3t(1-p)},\frac{\lambda}{3}}}$, combining these two bounds proves the Lemma.} 
\end{proof}
}

\section{Proofs for Mean- and Median-Based Algorithms}
\label{simple-algorithms}
In this section we provide the proofs for algorithms based on mean and medians. Throughout this section we assume that $n$ is at least \tcr{$14400$} for computational simplifications. 

\subsection{Upper Bounds for Mean and Median Estimators without Corruptions}
\label{sec:proof-without-corruptions}

\medskip
\noindent\textbf{Mean estimate.} The total number of edges in $ G\sim G(n,p)$ is a Binomial distribution with parameters ${n\choose 2}$ and $p$. Therefore, its expectation and variance are ${n\choose 2}p$ and ${n\choose 2}p(1-p)$. respectively. Thus, $\expectation{\hat p_{\rm mean}(G)}=p$ and $\Var(p_{\rm mean}(G))= p(1-p)/{n\choose 2}\le 4p(1-p)/n^2$.  
By Chebyshev's inequality, 
\[
\probof{\absv{\hat p_{\rm mean}(G)-p}\ge 20 \cdot \frac{\sqrt{p(1-p)}}{n} }\le 0.01.
\]



\noindent\textbf{Median estimate.} We will show that with probability at least $0.995$, the median degree of $G$ is at least $(n-1)p-C$ for some constant $C$.
The main hurdle in showing this is the fact that the node degrees $\deg(i)$ are not independent, which requires a careful analysis.
For $i\in[n]$, let  $ Y_i \coloneqq \id\Paren{\deg(i) \le p(n-1)-121}$. Then, $\sum_i Y_i$  is the number of nodes with degree at most  ${p(n-1)-121}$. 

We establish the following bounds for \tcr{$n\geq14400$}:
\begin{align}
&\expectation{\sum_i Y_i}\le \frac n2-15\sqrt{n}\label{eqn:mean-bound}\\
&\variance{\sum_i Y_i}\le n
\label{eqn:variance-bound}
\end{align}
With these, we can apply Cantelli's inequality to obtain:
\[
\probof{\sum_i Y_i\ge \frac n2}\le \frac{\variance{\sum Y_i}}{\variance{\sum Y_i}+(15\sqrt{n})^2}<0.005.
\]
This shows that with probability at least 0.995 the median degree is at least $ (n-1)p-121$.
By symmetry, with probability at least 0.995 the median degree is at most $(n-1)p+121$. By the union bound, with probability at least 0.99 the error of the median estimate is at most $121/(n-1)$.

We now prove~\eqref{eqn:mean-bound} and~\eqref{eqn:variance-bound} to complete the proof. 

To prove~\eqref{eqn:mean-bound}, note that $\deg(i)\sim \Bin(n-1,p)$ and $\expectation{Y_i} = \Pr[\Bin(n-1,p)\le p(n-1)-121 ]$.

We show that for any $n'$, $\Pr[\Bin(n',p)\le pn'-121 ]\le \frac 12-\frac {15}{\sqrt{n'+1}}$, then~\eqref{eqn:mean-bound} follows from the linearity of expectation. 
If $\Pr(\Bin(n',p)\le pn'-1 )\le \frac 12-\frac {15}{\sqrt{n'+1}}$ then we are done.
We prove for the case when $\Pr(\Bin(n',p)\le pn'-1 )\ge \frac 12-\frac {15}{\sqrt{n'+1}}$.
By Chebyshev's inequality, 
\[
\probof{\Bin(n',p)\le n'p-\sqrt{n}}\le \frac 1{4}.
\]
Then, for $n'\ge 14400 $,
\begin{align*}
\probof{\Bin(n',p)\in [n'p-\sqrt{n'}, pn'-1)} &= \Pr(\Bin(n',p)\le pn'-1 )- \probof{\Bin(n',p)\le n'p-\sqrt{n}}\\
&\ge \frac 12-\frac {15}{\sqrt{n'+1}}-\frac 1{4}\ge \frac18.
\end{align*}
Since the binomial distribution has a unique mode $\ge pn'-1$, then for any $t\le \sqrt{n'}$,
\[
\probof{\Bin(n',p)\in [n'p-t, pn'-1)} \ge \frac{t-1}{\sqrt {n'}-1}\cdot \frac18\ge  \frac{t-1}{\sqrt {n'+1}}\cdot \frac18.
\]
Since the median of $Bin(n',p)$ is $\ge n'p-1,$~\cite{KaasB80}, hence $\Pr[\Bin(n',p)\le pn'-1 ]\le 1/2$. From it subtracting the above equation for $t-1 = 15 \cdot 8=120$, we get $\Pr[\Bin(n',p)\le pn'-121 ]\le \frac 12-\frac {15}{\sqrt{n'+1}}$.

We now prove~\eqref{eqn:variance-bound}. Since $Y_i$'s are identically distributed indicator random variables,
\begin{align}\label{eqn:var-cov}
\variance{\sum_i Y_i} =n\variance{Y_1}+n(n-1)\covariance(Y_1, Y_2)\le \frac n4+n(n-1)\covariance(Y_1, Y_2).
\end{align}
Let $t=(n-1)p-121$, then  $Y_i= \id\Paren{deg(i) \le t}$. Let $Y_{12}$ be the number of edges from node 1 to $[n]\setminus \{2\}$ and $\id\Paren{E_{1,2}}$ be the indicator that edge between $1$ and $2$ is present. Then $Y_{12}\sim \Bin(n-2, p)$. Elementary computations using the observation that $Y_1 = \id\Paren{Y_{12} \le t-1}+ \id\Paren{Y_{12} = t}\cdot (1-\id\Paren{E_{1,2}})$ show that
\[
\covariance(Y_1, Y_2) = p(1-p)\cdot \probof{Y_{12}=t}^2.
\]
From Stirling's approximation at $t=np$, we have $\probof{Y_{12}=t}\le 1/\sqrt{\pi p(1-p)(n-2)}$, and therefore, 
\[
\covariance(Y_1, Y_2) = p(1-p)\cdot \probof{Y_{12}=t}^2\le \frac1{\pi(n-2)}\le \frac 1{3n}
\]
for $n>120^2$. Plugging this in~\eqref{eqn:var-cov} proves~\eqref{eqn:variance-bound}.

\subsection{Lower Bounds for Mean and Median Estimators under Corruptions}
\label{sec:mean-median-with-corruptions}

We will prove the $\gamma/2$ lower bound for the mean and median estimates. 
Consider the following oblivious adversary $\adver$.
\begin{itemize}
\item 
Pick a random subset $B\subset [n]$ of size $\corr n$.
\item
Let $\adver_1(G)$ be the graph obtained by adding all edges $(u,v)$ that have at least one node in $B$ to the graph $G$, and let $\adver_2(G)$ be the graph obtained by removing all edges that have at least one node in $B$ from the graph $G$. 
\item 
Output $\adver_1(G)$ or $\adver_2(G)$ chosen uniformly at random. 
\end{itemize}

Any node in $\adver_1(G)$ has degree at least $\gamma n$ more than the corresponding node in $\adver_2(G)$. Therefore, $\absv{\hat p_{\rm mean} (\adver_1(G)) - \hat p_{\rm mean} (\adver_2(G))}\ge \gamma$, and  $\absv{\hat p_{\rm med} (\adver_1(G)) - \hat p_{\rm med} (\adver_2(G))}\ge \gamma$. Therefore by the triangle inequality, with probability 0.5, 
$\absv{\hat p_{\rm mean} (\adver(G)) -p}\ge \gamma/2$, and  $\absv{\hat p_{\rm med} (\adver(G)) -p}\ge \gamma/2$.

\subsection{Upper Bounds for Prune-then-Mean/Median Algorithms}
\label{sec:prune-ub}
Recall the prune-then mean/median algorithm in Algorithm~\ref{algo:prune}. We remove $c\gamma$ fraction of nodes with the highest and lowest degrees, and then output the median (or mean) of the remaining subgraphs. We restate the performance bound of the algorithm here.

\prunethenmedianub*

\begin{proof}

Let $G\sim G(n,p)$. By Chernoff bound (Lemma~\ref{lem:chernoff}) and the union bound, with probability $\ge 1-1/n^2$, 
\[
\deg(i)\in \Paren{np - 100\sqrt{n{\log n}},np + 100\sqrt{n{\log n}}}
\]
for all nodes $i\in[n]$ of $G$. We condition on this event.

Suppose an adversary converts $G$ into $\adver(G)$ by corrupting nodes in $B\subset [n]$ with $|B|\le \gamma n$. Note that the degree of a node in $\goodnodes=[n]\setminus B$ cannot change by more than $\gamma n$. Therefore, for all nodes $i\in F$ in $\adver(G)$, 
\begin{align}\label{eqn:prthmdub}
\deg(i)\in\Paren{np - 100\sqrt{n{\log n}}-\corr n,np +100\sqrt{n{\log n}}+\corr n}.    
\end{align}
Therefore, at most $\corr n$ nodes do not satisfy~\eqref{eqn:prthmdub}. 
Since we remove $c\corr n$ nodes with the highest and the lowest degrees for $c\ge 1$ all such nodes are pruned. The degree of any node not pruned decreases by at most $2c\corr n$, and after pruning all degrees are in the following interval
\begin{align}\label{eqn:prunmeanub}
    \Paren{np - 100\sqrt{n{\log n}}-(2c+1)\corr n,np + 100\sqrt{n{\log n}}+\corr n}.
\end{align}
We can rewrite this interval as follows
\begin{align*}
    \Paren{n(1-2c\corr)p - 100\sqrt{n{\log n}}+(2cp-2c-1)\corr n,n(1-2c\corr)p + 100\sqrt{n{\log n}}+(2cp+1) \corr n}.
\end{align*}
The prune-then-median estimator outputs one of these degrees (normalized), and its error is at most
\[
 \Paren{\frac{100\sqrt{n\log n} + (4c+1)\corr n}{(1-2c\corr)n}}=\mathcal O\Paren{\sqrt{\frac{\log n}{n}}+c\gamma}.
\]

We now bound the performance of prune-then-mean estimator.
Let $V'\subseteq[n]$ be the nodes that are not pruned, so $|V'|=(1-2c\gamma)n$. Let $\goodnodes^p:=V'\cap F$ and $B^p:=V'\cap \bad$ be the uncorrupted and corrupted nodes that remain after pruning. We have $|B^p|\le |B|\le \gamma n$ and $ |\goodnodes^p|\ge (1-(2c+1)\corr)n$.


There are three types of edges among the nodes in $V'$: (i) $\mathcal{E}_1$: edges whose both end points are good nodes (in $F^p$), (ii) $\mathcal{E}_2$: edges with at least one end point in $B^p$. The mean estimator outputs
\[
\frac{|\mathcal{E}_1|+|\mathcal{E}_2|}{{|V'|\choose 2}}.
\]
Its error is at most
\begin{align*}
\absv{\frac{|\mathcal{E}_1|+|\mathcal{E}_2|}{{|V'|\choose 2}}-p}
&= \absv{\frac{|\mathcal{E}_1|-{{|F^p|\choose 2}}p}{{{|V'|\choose 2}}}}
+\absv{\frac{|\mathcal{E}_2|-(|V'|-|F^p|)((|V'|+|F^p|-1)/2)p}{{{|V'|\choose 2}}}}\\
&=\absv{\frac{|\mathcal{E}_1|-{{|F^p|\choose 2}}p}{{{|V'|\choose 2}}}}
+\absv{\frac{|\mathcal{E}_2|-|B^p|((|V'|+|F^p|-1)/2)p}{{{|V'|\choose 2}}}}
\end{align*}
We will bound each term individually.
Since the subgraph $F^p\times F^p$ between the uncorrupted nodes remains unaffected from the original graph $G$, then Theorem~\ref{th:mainconth} implies that, with probability $\ge 1-3n^{-2}$, 
\[
\absv{\frac{|\mathcal{E}_1|}{{|F^p|\choose 2}}-p}=\mathcal O\Paren{\max\left\{c\corr\sqrt{\frac{\ln(e/c\corr)}n}, \frac{c\gamma\log n}{n},\frac{1}{n} \right\}}\le \mathcal O\Paren{c\corr^2+\frac{\log n}{n}}.
\]
Therefore, 
\[
\absv{|\mathcal{E}_1|- {|F^p|\choose 2}p } = {|F^p|\choose 2} \cdot O\Paren{  c\corr^2+\frac{\log n}{n}} \le {|V'|\choose 2}\cdot  O\Paren{ c\corr^2+\frac{\log n}{n}} .
\]
This shows that the first error term is at most $O\Paren{c\corr^2+\frac{\log n}{n}}$.

We now consider the second term. 
Note that $|n-(|V'|+|F^p|-1)/2|\le 3 c\corr n$. By the triangle inequality,
\begin{align}
\absv{|\mathcal{E}_2|-\frac12 \cdot |B^p|(|V'|+|F^p|-1)p}
&\le \absv{|\mathcal{E}_2|-|B^p|\cdot np}+ 3\gamma n p\cdot|B^p|.\label{eqn:ub-te}
\end{align} 

Let $\deg(i)$ be the degree of  node $i$ after pruning. By the triangle inequality adding and subtracting $\sum_{i\in B^p}\deg(i)$ to the first term we obtain, 
\begin{align*}
\absv{|\mathcal E_2| - |B^p| \cdot np  } \le \absv{|\mathcal E_2| - \sum_{i\in B^p} \deg(i)}+\sum_{i\in B^p}| \deg(i)-np|.
\end{align*}

Now note that $\absv{\mathcal E_2}$ is the number of edges with at least one endpoint in $B^p$. Therefore $\absv{|\mathcal E_2|- \sum_{i\in B^p}\deg(i)}$ is the number of edges inside $B^p\times B^p$ and is at most $|B^p|^2$. For the second term we use the fact that each node in $B^p$ satisfies~\eqref{eqn:prunmeanub}, and $|B^p|\le \corr n$. This gives
\begin{align*}
\absv{|\mathcal E_2| - |B^p| \cdot np  } \le \absv{|\mathcal E_2| - \sum_{i\in B^p} \deg(i)}+\sum_{i\in B^p}| \deg(i)-np| \le |B^p| \cdot \Paren{100\sqrt{n{\log n}}+(2c+2)\corr n}.
\end{align*}

Plugging this along with the fact that $|B^p|\le \corr n$ in~\eqref{eqn:ub-te}, we obtain

\begin{align*}
\absv{|\mathcal{E}_2|-\frac12 \cdot |B^p|(|V'|+|F^p|-1)p}\le  \gamma n \cdot \left(\Paren{100\sqrt{n{\log n}}+(5c+2)\corr n} \right).
\end{align*} 
Since ${|V'|\choose 2} > (n/2)^2$, the second term can be bounded by
\[
\mathcal{O}\Paren{4\corr\cdot \Paren{\sqrt{\frac{\log n}n}+(5c+2)\corr}} = \mathcal{O}\Paren{c\corr^2+\frac{\log n}{n}},
\]
thus proving the result. 

\end{proof}
\subsection{Lower Bounds for Prune-then-Mean/Median Algorithms}
\label{sec:prune-lb}
We will prove the following result showing the tight dependence of the upper bounds on $\corr$.
\prunethenmedianlb*

Let $G\sim G(n,0.5)$. The oblivious adversary $\adver$ operates as follows.
It partitions $G$ into five random sets $B, S_0, S_1, S_2,$ and $S_3$ with $|B|=\corr n$, $|S_0|=c\corr n, |S_1|=c\corr n, |S_2|= \frac 23(1-(2c+1)\corr)n, |S_3|= \frac 13(1-(2c+1)\corr)n$. 

\begin{itemize}
\item 
Remove all edges with at least one endpoint in $B$.
\item 
Remove all edges between $S_0$ and $B$. 
\item 
Add all edges between $S_1$ and $B$.
\item
Connect each node in $B$ to each node in $S_2$ independently with probability $3/5$.
\item
Connect each node in $B$ to each node in $S_3$ independently with probability $3/10$.
\item
Connect nodes within $B$ to each other with probability $3/5$.
\end{itemize}

By the Chernoff bound (Lemma~\ref{lem:chernoff}) and the union bound, we obtain the following bounds on the node degrees in $\adver(G)$.
\begin{lemma}
In $\adver(G)$, the following hold with probability at least $1-3n^{-3}$
\begin{align*}
\deg(u)&=n\left(\frac12 +\frac{\corr}{10}\right)\pm 4\sqrt{n\log n}\quad & \text{for}~~ u\in  B,\\
\deg(u)&=n\cdot\left(\frac 12 -\frac {\corr}2\right)\pm 4\sqrt{n\log n}\quad & \text{for}~~ u\in S_0,\\
\deg(u)&=n\cdot\left(\frac 12 +\frac {\corr}2\right)\pm 4\sqrt{n\log n}\quad & \text{for}~~ u\in S_1,\\
\deg(u)&=n\cdot\left(\frac 12 +\frac {\corr}{10}\right)\pm 4\sqrt{n\log n}\quad & \text{for}~~ u\in S_2,\\
\deg(u)&=n\cdot\left(\frac 12 -\frac {\corr}{5}\right)\pm 4\sqrt{n\log n}\quad & \text{for}~~ u\in S_3.
\end{align*}
\end{lemma}

Since $\gamma>100\sqrt{{\log n}/{n}}$, the nodes in $S_0$ are the $c\corr n$ nodes with the lowest degrees and the nodes in $S_1$ are the $c\corr n$ nodes with the highest degrees, and they are pruned by the algorithm. Now since the sets $S_0$ and $S_1$ were randomly chosen ahead of time, in the pruned graph, once again by the Chernoff bound (Lemma~\ref{lem:chernoff}) and the union bound, the following holds with probability at least $1-3n^{-3}$
\begin{align*}
\deg(u)&=n\left(\frac {1-2c\corr}2 +\frac{\corr}{10}\right)\pm 8\sqrt{n\log n}\quad & \text{for}~~ u\in  B,\\
\deg(u)&=n\cdot\left(\frac {1-2c\corr}2 +\frac {\corr}{10}\right)\pm 8\sqrt{n\log n}\quad & \text{for}~~ u\in S_2,\\
\deg(u)&=n\cdot\left(\frac {1-2c\corr}2 -\frac {\corr}{5}\right)\pm 8\sqrt{n\log n}\quad & \text{for}~~ u\in S_3.
\end{align*}
Since we assume that $c\corr<0.25$, there are more nodes in $S_3$ than in $S_2\cup B$ and every node in $S_2\cup B$ had a higher degree than any node in $S_3$. Therefore a node in $S_3$ is chosen as the median node, thus deviating from the median degree by at least $\corr/5\pm 8\sqrt{{\log n}/{n}}>\corr/10$ for $\corr>100\sqrt{{\log n}/{n}}$. This proves the lower bound for prune-then-median estimate.

{Now for the prune-then-mean estimate, note that each edge that remains after pruning is chosen at random, independent of all other edges.} The total expected number of edges after pruning is \tcr{$\frac 12\cdot\frac{n^2(1-2c\corr)^2}{2}+\frac{n^2\corr^2}{20}$} and the variance is at most $n^2/4$. Therefore, the total error of the prune-then-mean estimate is at least $\corr^2/20\pm O(1/n)$, and since $\corr>100\sqrt{{\log n}/{n}}$, the error is at least $\corr^2/40$.

\section{Proof of Theorem~\ref{th:mainconth}}
\label{sec:app-mainconth}
Throughout this proof, let $\beta = \max\Big\{ 16\alpha n\sqrt{{pn}\ln\frac{e}{\alpha }},60\alpha n\ln\frac{e}{\alpha }, 5n\sqrt{p\ln (en)}\Big\}$. \tcr{First fix $\alpha\in[0,1/2]$.
}

We first consider the entire matrix $\tilde A$, namely $S=S'=[n]$. 
\tcr{Recall that the diagonal entries of $\tilde A$ are zero. Then}, note that $\sum_{(i,j)\in [n]\times [n]} (\tilde A_{i,j}-p) = {2\cdot}\sum_{\tcr{(i,j)\in [n]\times [n]:\ i>j}} (\tilde A_{i,j}-p)-np$. Now since all the entries $\tilde A_{ij}$ are independent for $i>j$, we can apply the Chernoff bound (Equation~\eqref{eqn:chernoff})  with $\lambda = \beta$ over these entries and with probability at least $1-n^{-3}$, 
\begin{align}
    \Bigg|\sum_{\tcr{(i,j)\in [n]\times [n]:\ i>j}} (\tilde A_{i,j}-p)\Bigg|\le \beta.
    \label{eq:event1}
\end{align}
\tcr{Since $np\le n\sqrt{p}\le \beta$, then from the above equation we get $|\sum_{(i,j)\in [n]\times [n]} (\tilde A_{i,j}-p)|\le 3\beta$, with probability at least $1-n^{-3}$. 
Note that for $\alpha < 1/n$ the statement only applies to $S = S' = [n]$, and thus this case is handled. In the remaining proof $\alpha\in[1/n,1/2]$.}

Conditioned on the event $|\sum_{(i,j)\in [n]\times [n]} (\tilde A_{i,j}-p)|\le 3\beta$, note that for all $T\subset [n]\times[n]$,   
\begin{align}\label{eq:follows}
\Bigg|\sum_{(i,j)\in T} (\tilde A_{i,j}-p)\Bigg|>6\beta \Rightarrow \Bigg|\sum_{(i,j)\in {T}^c} (\tilde A_{i,j}-p)\Bigg|>3\beta,
\end{align}
where ${T}^c = [n]\times[n]\setminus T$. 
\newhz{In particular, if $T=S\times S'$ with $|S|\ge n-\alpha n$ and $|S'|\ge n-\alpha n$, then $|{T}^c|<2\alpha \newhz{n^2}$ and if $\min\{|S|,|S'|\}\le\alpha n$, then $|T|\le\alpha \newhz{n^2}$. Therefore, for $T = S \times S'$ with $|S|,\,|S'| \in C_\alpha$, either $|T|$ or $|{T}^c|$ is smaller than $2\alpha \newhz{n^2}$. With this in hand, the theorem will follow from the following lemmas.}

\begin{lemma}
Let $T\subset [n]\times[n]$ be a given subset of size at most $2\alpha n^2$, then 
\[
\Pr\Bigg[ \Bigg|\sum_{(i,j)\in T} (\tilde A_{i,j}-p)\Bigg|\ge 3\beta \Bigg]\le 4\exp\Paren{-20\alpha n\ln{e/\alpha}}. 
\]
\label{lem:smallset}
\end{lemma}

We now bound the number of subsets of interest.
\begin{lemma}\label{lem:combbound}
For a given $\alpha\in[1/n,1/2]$, the number of sets $S, S'$ with $|S|, |S'|\in \setsize_\alpha$ is at most $4\exp(4\alpha n\ln(e/\alpha))$.
\end{lemma}

\tcr{For a given $\alpha\in[1/{n},1/2]$ and $T=S\times S'$ such that $|S|, |S'|\in \setsize_\alpha$, since either of $T$ or ${T}^c$ have size $\le 2\alpha n^2$, therefore, combining the two lemmas implies that with probability $\ge1-16\exp\Paren{-16\alpha n\ln{e/\alpha}}\ge 1-n^3$,
\[\min\Bigg\{\Bigg|\sum_{(i,j)\in T} (\tilde A_{i,j}-p)\Bigg|,\Bigg|\sum_{(i,j)\in {T}^c} (\tilde A_{i,j}-p)\Bigg|\Bigg\}\le 3\beta.
\]
Then from Equation~\eqref{eq:follows}, with probability $\ge 1-n^3-n^3$, 
$\Big|\sum_{(i,j)\in T} (\tilde A_{i,j}-p)\Big|\le 6\beta$. This completes the proof for a given value of $\alpha$. To extend it to all $\alpha\in [1/{n},1/2]$ first note that it suffices to prove the theorem for $\alpha\in \{\frac1n,\frac2n,...,\frac{\lfloor 0.5 n\rfloor}n\}$, and then upon taking the union bound over these values of $\alpha$ completes the proof.
} 

We now prove Lemma~\ref{lem:smallset}. Note that 
\begin{align}
\sum_{(i,j)\in T}(\tilde A_{i,j}-p) = \sum_{(i,j)\in T : i>j} (\tilde A_{i,j}-p) +\sum_{(i,j)\in T : i<j}(\tilde A_{i,j}-p) - \sum_{(i,i)\in T} p.
\end{align}

\tcr{Then using the triangle inequality, $\{(i,i)\in T\}\le n$ and $np\le \beta$ to disregard the third term (as done before),}
\[
\Pr\Bigg[ \Bigg|\sum_{(i,j)\in T} (\tilde A_{i,j}-p)\Bigg|\ge 2\beta \Bigg]\le 
\Pr\Bigg[ \Bigg|\sum_{(i,j)\in T : i>j} (\tilde A_{i,j}-p)\Bigg|\ge \beta \Bigg]+
\Pr\Bigg[ \Bigg|\sum_{(i,j)\in T : i<j} (\tilde A_{i,j}-p)\Bigg|\ge \beta \Bigg]. 
\]

\tcr{The two events on the right hand side are for sums of independent mean-centered Bernoulli random variables.} We will now apply the Chernoff bound (Equation~\eqref{eqn:chernoff}). Note that for a fixed $\lambda$ the right hand side of~\eqref{eqn:chernoff} is a non-decreasing function of $t$. Further note that $|\{(i,j)\in T : i>j\}|, |\{(i,j)\in T : i<j\}|\le |T|<2\alpha n^2$. Therefore,



\[
\Pr\Bigg[ \Bigg|\sum_{(i,j)\in T : i>j} (\tilde A_{i,j}-p)\Bigg|\ge \beta \Bigg]\le2\exp\Paren{-\min\Paren{\frac{\beta^2}{6\alpha n^2p},\frac{\beta}{3}}}\le 2\exp\Paren{-20\alpha n\ln\frac{e}{\alpha }}. 
\]
\tcr{Similarly,}
\[
\Pr\Bigg[ \Bigg|\sum_{(i,j)\in T : i<j} (\tilde A_{i,j}-p)\Bigg|\ge \beta \Bigg]\le 2\exp\Paren{-20\alpha n\ln\frac{e}{\alpha }}. 
\]
Combining the two bounds completes the proof of Lemma~\ref{lem:smallset}.

\newhz{
We finally prove Lemma~\ref{lem:combbound}.
The number of such sets can be upper bounded by $4\cdot  \Paren{\sum_{j=0}^{\lfloor \alpha\ns \rfloor} \binom{\ns}{j}}^2$, where
\begin{align*}
\sum_{j=0}^{\lfloor \alpha\ns \rfloor} \binom{\ns}{j} & \le (\alpha\ns  +1)\cdot \binom{\ns}{\lfloor \alpha\ns \rfloor} \le (\alpha \ns+1) \cdot \Paren{\frac{e}{\alpha} }^{\alpha\ns} \le e^{\alpha\ns\ln\Paren{\frac{e}{\alpha}}+\ln(\alpha\ns+1)  } \le e^{\alpha\ns\ln\Paren{\frac{e}{\alpha}}+\alpha\ns  } \le e^{2\alpha\ns \ln (e/\alpha)}. \nonumber
\end{align*}
}

\section{Missing Proofs from Section~\ref{sec:upper-bounds}}
\subsection{Proof of Lemma~\ref{lem:whatisinthename}}
\label{sec:whatisinthename}
First note that
\begin{align*}
 0 =  \sum_{i,j\in S}(A_{i,j}-p_S) = \sum_{i,j\in S\cap \goodnodes }(A_{i,j}-p_S)+ \sum_{i,j\in S\cap \goodnodes^c}(A_{i,j}-p_S)+ 2\sum_{i\in S\cap \goodnodes ,\, j\in S\cap \goodnodes^c}(A_{i,j}-p_S).
\end{align*}
Therefore,
\begin{align*}
\Bigg|\sum_{i,j\in S\cap \goodnodes }(A_{i,j}-p_S)\Bigg|\le  \Bigg|\sum_{i,j\in S\cap \goodnodes^c}(A_{i,j}-p_S)\Bigg|+ 2\Bigg|\sum_{i\in S\cap \goodnodes ,\, j\in S\cap \goodnodes^c}(A_{i,j}-p_S)\Bigg|.
\end{align*}
Hence,
\begin{align}
\frac{|\sum_{i,j\in S\cap \goodnodes }(A_{i,j}-p_S)|}{3}\le \max\Bigg\{\Bigg|\sum_{i,j\in S\cap \goodnodes^c}(A_{i,j}-p_S)\Bigg|,\Bigg|\sum_{i\in S\cap \goodnodes ,\, j\in S\cap \goodnodes^c}(A_{i,j}-p_S)\Bigg|\Bigg\}.\label{eqn:step-before-main}
\end{align}

From Lemma~\ref{lem:spsubmat} , Lemma~\ref{lem:specnormlb} and the above inequality, it follows that
\begin{align}
\|(A-p_S)_{S\times S}\|&\ge \max\Big\{\|(A-p_S)_{(S\cap \goodnodes^c)\times(S\cap \goodnodes^c)}\|,\|(A-p_S)_{(S\cap \goodnodes^c)\times(S\cap \goodnodes )}\|\Big\}\label{eqn:step-one-main} \\
&\ge \max\Big\{\frac{|\sum_{i,j\in S\cap \goodnodes^c }(A_{i,j}-p_S)|}{|S\cap \goodnodes^c|}, \frac{|\sum_{i\in S\cap \goodnodes^c,\,j\in S\cap \goodnodes }(A_{i,j}-p_S)|}{\sqrt{|S\cap \goodnodes |\cdot |S\cap \goodnodes^c|}}\Big\}\label{eqn:step-two-main}\\
&\ge \min\Big\{\frac{|\sum_{i,j\in S\cap \goodnodes }(A_{i,j}-p_S)|}{3|S\cap \goodnodes^c|}, \frac{|\sum_{i,j\in S\cap \goodnodes }(A_{i,j}-p_S)|}{3\sqrt{|S\cap \goodnodes |\cdot |S\cap \goodnodes^c|}}\Big\}\label{eqn:step-three-main}\\
&= \frac{|\sum_{i,j\in S\cap \goodnodes }(A_{i,j}-p_S)|}{3\sqrt{|S\cap \goodnodes |\cdot |S\cap \goodnodes^c|}}\cdot  \min\Big\{\sqrt{\frac{|S\cap \goodnodes |}{|S\cap \goodnodes^c|}}, 1\Big\}\notag\\
&= \frac{|p_{S\cap \goodnodes }-p_S\|S\cap \goodnodes |}{3}\cdot \min\Big\{ \frac{|S\cap \goodnodes |}{|S\cap \goodnodes^c|}, \sqrt{\frac{|S\cap \goodnodes |}{|S\cap \goodnodes^c|}}\Big\},\notag
\end{align}
where~\eqref{eqn:step-one-main} is from Lemma~\ref{lem:spsubmat},~\eqref{eqn:step-two-main} follows from Lemma~\ref{lem:specnormlb},~\eqref{eqn:step-three-main} from~\eqref{eqn:step-before-main}.

\subsection{Proof of Theorem~\ref{th:lacon}}
\label{sec:lacon}
Since eigenvalues of symmetric matrices are real, let $v\in \mathbb R^n$ be the normalized top eigenvector of $M$ with eigenvalue $\lambda\in \mathbb R$ such that $M v = \lambda v $ and $\|M\| = |\lambda|$. Since $Mv=\lambda v$, we have $M_{S\times [n]} \,v\ = \lambda v_S$, 
and  
\begin{align}
M_{S\times [n]} \,v\ = M_{S\times S}  \,v_S+ M_{S\times S^c} \,v_{S^c}\label{step:uno}
\end{align}
By Lemma~\ref{lem:triangleineq} on~\eqref{step:uno}, 
\begin{align}
&\norm{M_{S\times [n]} \,v} \le \norm{M_{S\times S}\,v_S}+ \norm{M_{S\times S^c}\,v_{S^c}}\notag\\
\Rightarrow\quad & { |\lambda| \cdot \norm{v_S} \le \rho |\lambda| \cdot \norm {v_S} + |\lambda| \cdot \norm{v_{S^c}}}\label{step:dos}\\
\Rightarrow\quad & (1-\rho)\norm{v_S} \le \norm{v_{S^c}}\notag\\
\Rightarrow\quad & (1-\rho)^2\norm{v_S}^2 \le \norm{v_{S^c}}^2 \notag
\end{align}
where~\eqref{step:dos} uses the assumption of the lemma.
Finally using $\norm{v_S}^2+ \norm{v_{S^c}}^2=1$ gives the bound. 

\subsection{An Approximate Top Eigenvector Suffices}
In this section, we prove a variant of Theorem~\ref{th:lacon}, which works with an approximate rather than an exact top eigenvector.
As discussed in Remark~\ref{rmk:running-time}, this allows us to use approximate top eigenvector procedures, reducing the runtime.

\begin{lemma}\label{lem:comp}
Let $M$ be a nonzero $n\times n$ real matrix such that for some set $S\subset [n]$ we have $\|M_{S\times S}\|\le 0.53\|M\|$.
Let $v\in \mathbb R^{ n}$ be a unit vector such that $\|M v\|\ge  0.99 \norm{M}$, then $\|v_{S^c}\|^2\ge \frac{1}{8}$.
\end{lemma}
\begin{proof}
Let $u = Mv$. 
Note that $M_{S\times [n]} \,v = u_S$ and $M_{S^c\times [n]} \,v=u_{S^c}$, therefore
\begin{align*}
    v^T\, M\, v=  v^T\, (M_{S\times [n]}+M_{S^c\times [n]})\, v=  v^T(u_S+u_{S^c})= 
    v_S^T\, u_S + v_{S^c}^T\, u_{S^c}.
\end{align*}
Then by the triangle inequality,
\begin{align}
    &|v^T\, M\, v| \le \norm{v_S}\cdot \norm{u_S} + \norm{v_{S^c}}\cdot \norm{u_{S^c}}\nonumber\\
     \Rightarrow\quad & 0.99 \norm{M}\le  \norm{v_S}\cdot \norm{u_S} + \norm{v_{S^c}}\cdot \norm{u_{S^c}}\notag\\
     \Rightarrow\quad & 0.99 \norm{M}\le \sqrt{1- \norm{v_{S^c}}^2}\cdot \norm{u_S} + \norm{v_{S^c}}\cdot \sqrt{{\norm M}^2-\norm{u_{S}}^2}.\notag
\end{align}
In the last line, we used the fact that $\norm{u}\le \norm M \cdot\norm v =\norm M$ and ${\norm u}^2=\norm{u_S}^2+\norm{u_{S^c}}^2$.
Rearranging this expression, it is easy to show that in the case  $\norm{u_S}^2\le \frac{3\norm{M}^2}{4}$, the inequality is violated if $||v_{S^c}\|^2\le \frac{1}{8}$. Therefore, $\norm{u_S}^2\le \frac{3\norm{M}^2}{4}$ implies  $||v_{S^c}\|^2\ge \frac{1}{8}$.

To prove the lemma, we must handle the remaining case: we show that if  $\norm{u_S}^2\ge \frac{3\norm{M}^2}{4}$, then $\|v_{S^c}\|^2\ge \frac{1}{8}$.

Note that
\begin{align}
M_{S\times [n]} \,v\ =  M_{S\times S}  \,v_S+ M_{S\times S^c} \,v_{S^c}.\notag
\end{align}

Then
\begin{align}
 &\norm{M_{S\times [n]} \,v} \le \norm{M_{S\times S}\,v_S}+ \norm{M_{S\times S^c}\,v_{S^c}}\notag\\
\Rightarrow\quad & {  \norm{u_S} \le 0.53  \norm{M} \cdot \norm {v_S} +  \norm{M}\cdot \norm{v_{S^c}}}\notag\\
\Rightarrow\quad & {  \norm{u_S}^2 \le 2 (0.53^2)  \norm{M}^2 \cdot \norm {v_S}^2 +  2\norm{M}^2\cdot \norm{v_{S^c}}}^2\notag\\
\Rightarrow\quad & {  \norm{u_S}^2 \le 0.5618 \norm{M}^2(1-\norm{v_{S^c}}^2) +2 \norm{M}^2 \cdot \norm{v_{S^c}}}^2\notag\\
\Rightarrow\quad & {  \norm{u_S}^2 \le 0.5618 \norm{M}^2 +1.4382 \norm{M}^2 \cdot \norm{v_{S^c}}}^2.\notag
\end{align}
When $\norm{u_S}^2\ge 3\norm{M}^2/4$, the above equation implies $\norm{v_{S^c}}^2\ge 1/8$, which completes the proof of the lemma.
\end{proof}

\end{document}